\documentclass[lettersize,journal]{IEEEtran}
\IEEEoverridecommandlockouts
\usepackage{cite}
\usepackage{amsmath,amssymb,amsfonts}
\usepackage{algorithmic}
\usepackage{graphicx}
\usepackage{textcomp}
\usepackage{xcolor}
\def\BibTeX{{\rm B\kern-.05em{\sc i\kern-.025em b}\kern-.08em
    T\kern-.1667em\lower.7ex\hbox{E}\kern-.125emX}}

\usepackage[utf8]{inputenc} 
\usepackage[T1]{fontenc}
\usepackage{url}
\usepackage{ifthen}
\newcommand{\R}{\mathbb{R}}
\newcommand{\E}{\mathbb{E}}

\usepackage{algorithm}
\usepackage{amsthm}
\theoremstyle{definition}
\newtheorem{definition}{Definition}
\newtheorem{remark}{Remark}
\theoremstyle{plain}

\newtheorem{theorem}{Theorem}
\newtheorem{lemma}{Lemma}
\newtheorem{corollary}{Corollary}

\usepackage{array}
\usepackage{makecell}
\usepackage{multicol}
\usepackage{multirow}
\usepackage{threeparttable}
\usepackage{stfloats}

\usepackage{subfigure}

\usepackage[nameinlink]{cleveref}
\crefformat{equation}{(#2#1#3)}

\usepackage{mathtools}

\begin{document}

\title{Robustness in Wireless Distributed Learning: An Information-Theoretic Analysis}

\author{Yangshuo~He,~Guanding~Yu,~and~Huaiyu~Dai%
\thanks{This work was supported by the National Natural Science Foundation Program of China under Grant 62471434.}
\thanks{Yangshuo He and Guanding Yu are with the College of Information Science and Electronic Engineering, Zhejiang University, 38 Zheda Road, Hangzhou, China, 310027, email: \{sugarhe@zju.edu.cn, yuguanding@zju.edu.cn\} \textit{(Corresponding author: Guanding Yu.)}}
\thanks{Huaiyu Dai is with the Department of Electrical and Computer Engineering, North Carolina State University, Raleigh, NC 27695 USA, e-mail: \{hdai@ncsu.edu\}.}
}

\maketitle

\begin{abstract}
In recent years, the application of artificial intelligence (AI) in wireless communications has demonstrated inherent robustness against wireless channel distortions. Most existing works empirically leverage this robustness to yield considerable performance gains through AI architectural designs. However, there is a lack of direct theoretical analysis of this robustness and its potential to enhance communication efficiency, which restricts the full exploitation of these advantages. In this paper, we adopt an information-theoretic approach to evaluate the robustness in wireless distributed learning by deriving an upper bound on the task performance loss due to imperfect wireless channels. Utilizing this insight, we define task outage probability and characterize the maximum transmission rate under task accuracy guarantees, referred to as the task-aware $\epsilon$-capacity resulting from the robustness. 
To achieve the utility of the theoretical results in practical settings, we present an efficient algorithm for the approximation of the upper bound. Subsequently, we devise a robust training framework that optimizes the trade-off between robustness and task accuracy, enhancing the robustness against channel distortions. Extensive experiments validate the effectiveness of the proposed upper bound and task-aware $\epsilon$-capacity and demonstrate that the proposed robust training framework achieves high robustness, thus ensuring a high transmission rate while maintaining inference performance.
\end{abstract}


\section{Introduction}

As artificial intelligence (AI) applications are drawing intensive attention in wireless communication systems, wireless distributed learning emerges as a key characteristic of future communications.
In wireless distributed learning, intelligent edge devices in wireless networks collaborate to achieve AI tasks by collecting, transmitting, and processing task-related data. 
The power of AI has led to a paradigm shift from conventional bit-based metrics to task-aware metrics and requires novel approaches to enhance communication efficiency.
In such a scenario, focusing on AI task performance enables the transmitter and receiver to convey information that is both relevant to the task and robust to channel distortion. Hence, robust communication in terms of AI applications can be achieved even under lossy transmission.
It has been shown that exploiting such robustness against wireless channel noise provides a promising method to improve communication efficiency \cite{IB}. Moreover, the robustness in wireless distributed learning reveals the potential to transcend the limitation set by Shannon channel capacity.
This fundamental evolution has stimulated various AI-based joint communication and learning designs that aim at robust and efficient communication in wireless distributed learning \cite{JSCC-text,JSCC-image,DeepSC}. 
However, although many studies on architectural designs demonstrated the effectiveness of implicitly utilizing the inherent robustness of AI to enhance transmission rate, little attention has been paid to the theoretical analysis of such robustness, which hinders further development in this promising field. Therefore, it is essential to comprehend and quantify the inherent robustness in wireless distributed learning, thereby enabling the design of innovative and explainable methods to improve AI task performance.

\subsection{Related Works}
Various joint communication and learning designs emerge in light of the robustness in neural networks (NNs), such as the deep joint source-channel coding architecture for wireless image transmission in \cite{JSCC-image}, where the authors employed the NN-based encoder and decoder to jointly optimize the image reconstruction quality under wireless channel noise. Such a joint design was also applied in natural language processing for text transmission \cite{JSCC-text}, which jointly trained the NNs for both sentence similarity and transmission robustness. Beyond the goal of reconstruction in data-oriented communication, the authors in \cite{IB} introduced the joint communication and learning technique to task-aware communication that focuses on the accuracy of the AI tasks.
However, it is worth noting that most of these works focus on the architecture designs of NNs to achieve performance gains, primarily relying on empirical methods rather than theoretical foundations.

Some previous works attempted to theoretically analyze the performance gain of wireless distributed learning attributed to its robustness. 
By mimicking Shannon's theory, the authors in \cite{towards-semantic-theory} proposed a theoretical framework for semantic-level communication based on logical probability.
Although the analysis is limited to the logical probability system, it sheds light on the theoretical research of integrating AI into wireless distributed learning.
Adopting such an idea, the authors in \cite{semantic-capacity} proposed the semantic channel capacity benefiting from NNs' robustness, which transcends the Shannon capacity. But there was no analytical solution and only numerical approximation by simulations was provided.
To understand and exploit the intrinsic robustness within NNs, the authors in \cite{IB} utilized the information bottleneck (IB) framework to characterize the trade-off between communication overhead and inference performance. The robustness is implicitly measured by the mutual information relevant to the task accuracy.
A rate-distortion framework was proposed in \cite{RDT} for the information source in AI applications. This framework characterizes the trade-off between the bit-level/task-level distortions and code rate.
Furthermore, the authors in \cite{semantic_theory_Deniz} delved into the applications of the language models in wireless communication systems. A distortion-cost region was established as a critical metric for assessing task performance. 
In our previous work \cite{multi-modal-robustness}, we also investigated the inherent robustness in a multi-modal scenario. By employing the robustness verification problem, we evaluate the negative effect of wireless channels on task performance and thus measure the importance of each modality.

\subsection{Motivations and Contributions}
Integrating AI into wireless networks has shown great potential in mitigating channel distortions and enhancing transmission rates. Nevertheless, the absence of interpretable mathematical analysis on robustness may restrict the full utilization of these advancements to improve communication efficiency and task accuracy in wireless distributed learning.
One attempt to evaluate the robustness in wireless distributed learning is to introduce information theory into AI tasks. For example, \cite{IB} utilized mutual information to inherently measure the task-aware robustness and incorporated the IB framework to improve the robustness against channel distortion.
However, the mutual information trade-off in \cite{IB} only provides an implicit evaluation of the robustness, which limits the effectiveness of this approach in improving robustness and leveraging its potential advantages for more aggressive transmission.
Motivated by these limitations, we employ information theory to directly examine the robustness in wireless distributed learning. This paper aims to derive a theoretical measure for the inherent robustness of a learning system against wireless distortions and shed light on its potential benefit of enabling a transmission rate higher than channel capacity. 
To fully exploit the advantages of this robustness, we further propose a robust training framework for wireless distributed learning. Our goal is to enhance the intrinsic robustness and improve communication efficiency, thereby maximizing the benefits of AI in wireless distributed learning.
The main contributions are summarized as follows:

\begin{itemize}
    \item We evaluate the inherent robustness in wireless distributed learning by deriving an upper bound on the performance degradation due to the imperfect wireless channel. Specifically, we leverage the loss function as a metric for task performance, which is generally applicable to various AI tasks such as classification and regression. The AI's capability of mitigating the negative effect of wireless channels could be inferred from its impact on the loss function. By incorporating the Donsker-Varadhan representation of the Kullback-Leibler (KL) divergence, we further establish an upper bound on the performance degradation.
    \item To leverage the robustness in wireless distributed learning for enhancing transmission rates, we propose the task-aware $\epsilon$-capacity, defined under the conditions ensuring successful inference. In particular, we introduce a concept of task outage based on the derived upper bound to characterize instances of failure in wireless distributed inference. Building on this framework, we further define the task-aware $\epsilon$-achievable rate which captures the achievability of a transmission rate in terms of task accuracy. By deriving an upper bound on this rate, termed the task-aware $\epsilon$-capacity, we highlight the communication efficiency gains attributed to the inherent robustness in wireless distributed learning. This theoretical gain substantiates the potential of utilizing robustness to increase transmission rate in a task-aware manner.
    \item Given that the mutual information term in the derived upper bound is often intractable for most AI tasks with high dimensional features, we provide a numerical estimation for it in the context of NNs. Under common assumptions of NNs, we employ the bootstrapping technique to calculate the mutual information and develop an efficient algorithm for its estimation. This approach allows for the practical application of the previously proposed robustness analysis. Specifically, we apply this approach to a wireless distributed learning system, demonstrating that the performance loss due to the imperfect wireless channel is upper bounded by the derived theoretical bound and showing its usage in the task outage probability and task-aware $\epsilon$-capacity in a realistic setting.
    \item We propose a robust training framework that leverages the derived upper bound on the task performance degradation to improve the inherent robustness in wireless distributed learning. To this end, we design an objective function that balances the trade-off between robustness and AI task performance. Our formulation aims to minimize the adverse impact of imperfect wireless channels while maximizing task accuracy. The proposed framework ensures the model achieves high task-aware robustness and is capable of transmitting data at a rate beyond the capacity limit while maintaining the inference performance.
\end{itemize}

\subsection{Notations and Organization of the Paper}
We denote the upper-case letters as random variables, such as $X$, lower-case letters as their realizations, such as $x$, and calligraphic letters as their sets, such as $\mathcal
X$. The expected value and the probability density function (pdf) of the random variable $X$ are denoted by $\E_{p(x)}\left[X\right]$ and $p(x)$, respectively. We use $\Pr\left\{X\right\}$ to denote the probability that $X$ is true. The entropy of $X$ is denoted by $H(X)$ or $H(p(x))$. We use $\mathcal{N}(\mu,\sigma^2)$ for the Gaussian distribution with mean $\mu$ and variance $\sigma^2$ and $\mathcal{N}(x|\mu,\sigma^2)$ for its pdf. The zero vectors and identity matrices are denoted by $\mathbf{0}$ and $\mathbf{I}$. We denote the transpose and absolute value by $(\cdot)^\top$ and $\left|\cdot\right|$.

The rest of this paper is organized as follows: Section \ref{sec:system_model} introduces the system model of wireless distributed learning and formulates the performance loss resulting from the imperfect wireless channel. 
Section \ref{sec:main_results} presents the theoretical analysis of the inherent robustness in wireless distributed learning, including the derived upper bound on the performance loss and the channel capacity benefiting from this robustness. Section \ref{sec:numerical_example} provides a numerical example of the theoretical results within the context of NNs, validating their effectiveness. Section \ref{sec:robust_training_framework} proposes the robust training framework based on the theoretical analysis. Numerical simulations are summarized in Section \ref{sec:simulation}. Finally, Section \ref{sec:conclusion} concludes this paper.

\section{System Model}
\label{sec:system_model}
As shown in Figure \ref{fig:system_model}, we consider a standard setting of distributed learning in the context of wireless communication systems with one base station (BS) and one user equipment (UE). 
Let $\mathcal{Z}=\mathcal{X}\times\mathcal{Y}$ denote the sample space, and each sample $Z\in\mathcal{Z}$ is a pair of input data $X$ and target variable $Y$, i.e., $Z=(X,Y)$. A learning algorithm $p(w|z)$ maps the sample $Z$ to the hypothesis $W\in\mathcal{W}$. 
The learning model $f^{(W)}:\mathcal{X}\to\mathcal{Y}$ outputs the target variable $Y$ given input data $X$\footnote{In neural network settings, the function of the learning model $f^{(W)}$ represents the network architecture, while hypothesis $W$ represents the weights.}. Within the setting of wireless distributed learning, the model is further separated into an encoder $f_e^{(W)}$ and a decoder $f_d^{(W)}$, deployed at the UE and BS respectively. 
The data $X$ is encoded by the encoder into signal $M$ and transmitted to the decoder over a wireless channel with side information $S$. In this study, we consider the channel state as side information, such as the channel state information in a practical wireless communication system. Specifically, we denote the pdf of the side information as $\mu(s)$. Each realization $s\in\mathcal{S}$ can characterize a specific channel condition.
The decoder at the BS leverages the received signal $\hat{M}$ for decoding and outputs the target result $\hat{Y}$. These random variables constitute the following Markov chain:
\begin{equation}
    \label{equ:markov_chain}
    Y\to X\to M\to \hat{M} \to \hat{Y}.
\end{equation}
In a standard distributed learning model without wireless communication, the Markov chain can be simplified as $Y\to X\to M\to \Tilde{Y}$.

\begin{figure}[htbp]
  \centering
  \includegraphics[width=0.9\columnwidth]{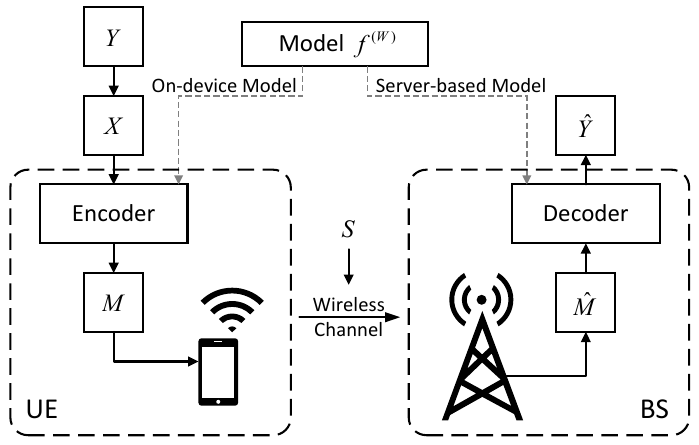}
  \caption{System model.}
  \label{fig:system_model}
\end{figure}

The performance of the learning model $f^{(W)}$ can be generally measured by a loss function $l(f^{(w)}(x),y)$ that calculates the distance between the model output and the target variable.
The loss function is considered as the performance measure for generality, as it would be chosen to match different tasks, including regression, classification, and data recovery in an end-to-end data-oriented communication scenario.
We formally define the loss function as $l:\mathcal{W}\times \mathcal{Z}\to \R^+$. 
First, we consider the performance of standard distributed learning. 
We denote the loss function for a model output $\tilde{y}$ as standard loss $\Tilde{l}\triangleq l(\Tilde{y}, y)$. Given the joint distribution of $p(w,z)$, the \textit{standard risk} is defined as
\begin{equation}
    \label{equ:true_risk}
    L \triangleq \E_{p(w|z)}\E_{p(x,y)}\left[l(W,Z)\right] = \E_{p(w,z)}\left[l(W,Z)\right].
\end{equation}
Note that in the above expression, the expectation is also taken over $p(w|z)$, as we are evaluating the risk on all learned posteriors instead of a specific hypothesis $w$. 

In the context of wireless communication systems, the noisy channel may affect the feature transmission and learning algorithm.  
Specifically, in the transmission process, the received signal $\hat{M}$ may not be identical to the transmitted signal $M$ due to the imperfect wireless channel. Meanwhile, based on the given side information $S$, the BS recovers hidden layer features from the received signal $\hat{M}$. The decoder $f_d^{(W)}$ utilizes these features to compute target results for the training algorithm. This enables the hypothesis $W$ in the decoder to implicitly learn information about the wireless channel in the training process.
Therefore, the wireless channel can affect the model output $\hat{y}$ and influence the loss function, which we denote as wireless loss $\hat{l}\triangleq l(\hat{y},y)$. In this scenario, we consider the joint distribution of the hypothesis and sample space conditioned on the side information $q(w,z|s)$. 
Given this conditional distribution, the \textit{wireless risk} of a specific channel condition $s$ is defined as
\begin{equation}
    \label{equ:wireless_risk}
    L(s) \triangleq \E_{q(w, z | s)}\left[l(W,Z)\right].
\end{equation}
Therefore, the absolute difference between the wireless risk and the standard risk can be considered as the loss difference in the wireless environment. We denote the expectation over all channels as \textit{wireless risk discrepancy}, which is given by
\begin{equation}
    \label{equ:wireless_discrepancy}
    g(\mu,f) \triangleq \E_{\mu(s)}\left[| L(S) - L|\right],
\end{equation}
where $\mu=\mu(s)$ is the probability distribution of the side information and $f=f^{(\cdot)}$ is the function of the learning model.
The definitions are organized into Table \ref{tab:definitions} for better clarification.

\begin{table}[htbp]
  \caption{Definitions}
  \label{tab:definitions}
  \centering
  \begin{tabular}{| m{0.12\columnwidth}<{\centering} | m{0.76\columnwidth}<{\centering} |}
    \hline
    Term & Definition \\ \hline
    $Z$ & Sample drawn from space $\mathcal{Z}$ \\ \hline
    $X$, $Y$ & Input data and target variable \\ \hline
    $W$ & Hypothesis \\ \hline
    $f^{(W)}$ & Learning model given hypothesis $W$ \\ \hline
    $M,\hat{M}$ & Transmitted and received signals \\ \hline
    $S$ & Side information of a wireless channel \\ \hline
    $\mu$ & Probability distribution of the side information \\ \hline
    $\tilde{Y},\hat{Y}$ & Target results of standard and wireless distributed learning\\ \hline
    $p(w,z)$ & Joint distribution of input data and target variable \\ \hline
    $q(w,z|s)$ & Conditional distribution given side information \\ \hline
    $\tilde{l},\hat{l}$ & Standard loss and wireless loss \\ \hline
    $L$ & Standard risk \\ \hline
    $L(s)$ & Wireless risk given side information $s$ \\ \hline
    $g(\mu,f)$ & Wireless risk discrepancy given $\mu$ and $f$ \\ \hline
    $D_{\text{KL}}(\cdot\|\cdot)$ & Kullback-Leibler divergence \\ \hline
    $I(A;B)$ & Mutual information between variables $A$ and $B$ \\ \hline
    $H(A)$ & Entropy of variable $A$ with distribution $p(a)$ \\ \hline
    $CE(\rho,\pi)$ & Cross-entropy between distributions $\pi$ and $\rho$ \\ \hline
    $G$ & Upper bound of the wireless risk discrepancy $g(\mu,f)$ \\ \hline
    $C_\epsilon$ & Task-aware $\epsilon$-capacity \\ \hline
  \end{tabular}
\end{table}

The discrepancy can be viewed as a characterization of the model's robustness to the imperfect wireless channel and its meaning is twofold.
First, via the power of deep learning, we can learn a model minimizing the difference to achieve better robustness against distortion in wireless communication systems. 
Second, a small discrepancy indicates that the model allows for more channel perturbation while maintaining task performance, thus a relatively high transmission rate can be adopted. 

Traditional communication systems generally target error-free transmission or very small bit error rate (BER). 
However, error-free transmission does not necessarily align with the objectives of AI in a wireless scenario.
The main goal of wireless distributed learning is to execute the task with sufficient precision, e.g., as measured by the accuracy in classification and the mean absolute error in regression.
In the presence of wireless channel distortion, a task may still be completed successfully given the robustness of learning models. 
Therefore, we can transition from error-free to task-reliable transmission, exploiting the possible capacity gain in a wireless distributed learning system. 
To meet the task-reliable criterion, the system must ensure that the wireless risk discrepancy is negligible for the task's successful execution despite transmission errors.

\section{Theoretical analysis of the robustness}
\label{sec:main_results}
In this section, we incorporate information theory to evaluate the robustness against channel distortions in wireless distributed learning. By investigating the discrepancy in loss functions, we derive an upper bound indicating the ability to mitigate the impact of channel noise. Subsequently, we analyze the benefit of such robustness in improving communication efficiency. In a task-aware manner, we introduce a novel concept of task outage and thus further propose the task-aware $\epsilon$-capacity benefiting from the robustness in wireless distributed learning. 

\subsection{Upper Bound of Wireless Discrepancy}
\label{sec:upper_bound}
To carry out information-theoretic analysis, we revisit the wireless risk discrepancy $g(\mu,f)$ in \cref{equ:wireless_discrepancy}, which measures the performance deterioration caused by the imperfect wireless channel. It can be viewed as the expected difference between distributions $q(w, z | s)$ and $p(w,z)$.
Therefore, KL divergence $D_{\text{KL}}(\cdot\|\cdot)$ serves as a useful tool to measure the distance between $q$ and $p$.
We leverage the Donsker-Varadhan variational representation of KL divergence in \cite{KL-divergence-variant} as the following lemma:
\begin{lemma}
\label{thm:KL_divergence}
For any two probability measures $\pi$ and $\rho$ defined on a common measurable space $(\Omega,\mathcal{F})$ where $\mathcal{F}$ is a $\sigma$-algebra on the set $\Omega$. Suppose that $\pi$ is absolutely continuous with respect to $\rho$, i.e., for every event $A$, if $\pi(A)=0$ then $\rho(A)=0$. Then
\begin{equation}
    \label{equ:KL_divergence}
    D_{\text{KL}}(\pi\|\rho) = \sup_{F\in\mathcal{F}}\left\{\int_\Omega F d\pi - \log \int_\Omega e^F d\rho\right\},
\end{equation}
where the supremum is taken over all measurable functions $F$, such that $F$ is integrable under $\pi$ and $e^F$ is integrable under $\rho$.
\end{lemma}

We further assume that the loss function $l(w,z)$ is $\sigma$-sub-Gaussian\footnote{A random variable $X$ is $\sigma$-sub-Gaussian on $p$ if for all $\lambda\in\R$, $\log\E_p\left[e^{\lambda X}\right]\leq \lambda^2\sigma^2/2$.} on distribution $p(w,z)$, which can be readily fulfilled by using a clipped loss function \cite{sub-Gaussian} \cite{information-theoretic}. 
By letting $\pi=q(w, z | s)$, $\rho=p(w,z)$, and $F = \lambda l(W,Z)$, we can derive the following theorem:

\begin{theorem}
    \label{thm:upper_bound}
    If the loss function $l(w,z)$ is $\sigma$-sub-Gaussian on distribution $p(w,z)$, then the wireless risk discrepancy is upper bounded by
    \begin{equation}
        \label{equ:upper_bound}
        g(\mu,f) \leq \sigma \sqrt{2I(W,Z;S)} \triangleq G.
    \end{equation}
\end{theorem}

\begin{proof}
    Let a measurable function be $F=\lambda l(W,Z),\ \forall \lambda\in\R$, given the supremum in \cref{equ:KL_divergence}, we have the following inequality:
    \begin{subequations}
        \label{equ:KL_divergence_gap}
        \begin{align}
            D_{\text{KL}}(q\| p) & \geq \E_{q}\left[\lambda l(W,Z)\right] - \log \E_{p}\left[e^{\lambda l(W,Z)}\right] \label{equ:KL_divergence_gap-1} \\
            & \geq \E_{q}\left[\lambda l(W,Z)\right] - \frac{\lambda^2\sigma^2}{2} \label{equ:KL_divergence_gap-2} \\
            & \geq  \lambda\left(\E_{q}\left[ l(W,Z)\right]- \E_{p}\left[ l(W,Z)\right]\right) - \frac{\lambda^2\sigma^2}{2}, \label{equ:KL_divergence_gap-3}
        \end{align}
    \end{subequations}
    where \cref{equ:KL_divergence_gap-1} is from the supremum, \cref{equ:KL_divergence_gap-2} follows from the sub-Gaussian assumption of the loss function $l(w,z)$ on distribution $p(w,z)$, and \cref{equ:KL_divergence_gap-3} is from the non-negativity of $\E_{p}\left[ l(W,Z)\right]$.
    For simplicity, we denote $D_{\text{KL}}(q(w, z | s)\| p(w,z))$ as $D(q\|p)$. 
    Inequality \cref{equ:KL_divergence_gap-3} indicates a non-negative parabola in $\lambda$ as follows:
    \begin{equation}
        \label{equ:parabola}
        \frac{\sigma^2}{2}\lambda^2 - \left(\E_{q}\left[ l(W,Z)\right] - \E_{p}\left[ l(W,Z)\right]\right)\lambda + D(q\|p) \geq 0.
    \end{equation}
    To ensure the non-negativity of the parabolic function, a non-positive discriminant is required
    \begin{equation}
        \label{equ:KL_divergence_discriminant}
        \frac{1}{2 \sigma^2}\left(\E_q[l(W, Z)]-\E_p[l(W, Z)]\right)^2 - D(q\|p) \leq 0.
    \end{equation}
    Taking the expectation of the square root under the side information distribution $\mu$ gives
    \begin{equation}
        \label{equ:jensen_inequality}
        \begin{split}
            \E_{\mu}\left[\left|\E_q[l(W,Z)]-\E_p[l(W,Z)]\right|\right] & \leq \E_{\mu}\left[\sqrt{2 \sigma^2 D(q\|p)}\right] \\
            & \leq \sqrt{2 \sigma^2 \E_{\mu}\left[D(q\|p)\right]},
        \end{split}
    \end{equation}
    where the second inequality follows from Jensen's inequality on the concave square root function. Note that the mutual information is the expectation of the KL divergence $I(W,Z;S) = \E_{\mu(s)}\left[ D(q\|p)\right]$, we can derive $g(\mu,f) \leq \sigma \sqrt{2I(W,Z;S)}$. 
\end{proof}

Theorem \ref{thm:upper_bound} provides an upper bound on the expected discrepancy between the wireless risk and the standard risk. This discrepancy evaluates the deterioration of task performance caused by the wireless channel distortion. A small discrepancy indicates the model is relatively robust to the channel noise. Since the wireless risk discrepancy in \cref{equ:wireless_discrepancy} is intractable during the training and inference processes of wireless distributed learning, it is natural to utilize the derived upper bound $G=\sigma\sqrt{2I(W,Z;S)}$ as an alternative measure of task-aware robustness. Minimizing the upper bound enables the model to transmit data at a high rate beyond capacity while preserving task accuracy. Therefore, it is essential for the wireless distributed learning model to learn an optimal hypothesis that achieves a tight upper bound.
In the following section, we will derive an efficient algorithm to estimate the upper bound in the NN scenario and propose a novel framework for robust inference in wireless distributed learning.

\begin{remark}
    \label{thm:special_case}
    It is worth noting that the wireless risk discrepancy becomes zero upon achieving error-free transmission, indicating exact recovery of the transmitted features, which is the goal of data-oriented communication.
    Additionally, if the mutual information $I(W,Z;S) = 0$, the wireless risk discrepancy also becomes zero. 
    It suggests that $(W,Z)$ are jointly independent of side information $S$. In this scenario, the model generalizes to any wireless channel as analyzed in \cite{PAC-Bayes}, thus accomplishing perfect inference in wireless distributed learning.
    However, absolute independence is virtually unattainable. In practice, the learning algorithm can only find a model where the mutual information approaches zero.
\end{remark}

\begin{remark}
    \label{thm:chain_rule}
    Using the chain rule of mutual information, $I(W,Z;S)$ can be written as
    \begin{equation}
        \label{equ:MI_first}
        I(W,Z;S) = I(Z;S) + I(W;S|Z),
    \end{equation}
    and
    \begin{equation}
        \label{equ:MI_second}
        I(W,Z;S) = I(W;S) + I(Z;S|W).
    \end{equation}
    From \cref{equ:MI_first}, it can be observed that the sample $Z$, which consists of both input data $X$ and target $Y$, is independent of the wireless channel. Hence, the first term $I(Z;S)$ becomes zero, and the mutual information reduces to $I(W,Z;S) = I(W;S|Z)$.
    The first term in \cref{equ:MI_second} represents the information about the channel contained in the learned hypothesis. From the perspective of probably approximately correct (PAC) \cite{PAC-Bayes}, a small $I(W;S)$ mitigates over-fitting on a specific channel condition and thus increases robustness against wireless channels. The second term can be further written as $I(Z;S|W) = I(Y;S|W,X) + I(X;S|W)$, where $I(Y;S|W,X) = \E_{p(x,w,s)}\left[D_{\text{KL}}(p(y|x,w,s)\|p(y|x,w))\right]$. 
    In the classification problem, the cross-entropy between the wireless prediction $p(y|x,w,s)$ and the standard prediction $p(y|x,w)$ becomes $CE\left(p(y|x,w,s), p(y|x,w)\right) = H(p(y|x,w,s)) + D_{\text{KL}}(p(y|x,w,s)\|p(y|x,w))$. Therefore, minimizing the mutual information $I(W,Z;S)$ is equivalent to reducing the distance between $\hat{y}$ and $\Tilde{y}$, thereby mitigating the negative effect of the wireless channel.
\end{remark}

\begin{remark}
    In this part, we examine the achievability of the upper bound presented in Theorem 1. The sub-Gaussian condition in \cref{equ:KL_divergence_gap-2} takes an equal sign when the variable $l$ follows a normal distribution under $p(w,z)$. 
    The equality in \cref{equ:KL_divergence_gap-3} is achieved when the standard risk is zero, i.e., $\E_p[l(W, Z)]=0$. This condition implies that the model achieves optimality in standard distributed learning. For instance, a zero mean squared error in regression and an exact match between output and label distribution.
    However, it represents an idealized scenario that is generally unattainable in practical applications. Therefore, this equality only serves as a conceptual benchmark rather than an achievable setting in the real world.
    It is further discerned from \cref{equ:KL_divergence_discriminant} that when the discriminant is zero, there is a single solution for the parabola, given by $\lambda_0=\sqrt{2D_{\text{KL}}(q\|p)/\sigma^2}$. 
    Finally, Jensen's inequality in \cref{equ:jensen_inequality} becomes trivial when $D(q\|p)$ is constant. In this scenario, the two distributions $q(w,z|s)$ and $p(w,z)$ either are identical or have a fixed relationship, indicating perfect inference as analyzed in Remark \ref{thm:special_case}.
    Therefore, the upper bound of the wireless performance gap in Theorem \ref{thm:upper_bound} is achieved when the model becomes perfect in both standard and wireless distributed learning and the loss function adheres to a normal distribution $l(W,Z)\sim \mathcal{N}(0, \sigma^2)$ under $p(w,z)$. 
    However, due to the intractability of $l(w,z)$, deriving the analytical forms of $q(w,z|s)$ and $p(w,z)$ that satisfy the upper bound equality is not feasible in practical simulations.
    We can only approach the upper bound by optimizing $q(w,z|s)$ in the training process.
\end{remark}

\begin{remark}
    The existing information-theoretic analysis based on IB \cite{IB} utilizes $I(X;T)$ to measure the robustness of learning models in edge inference inherently. Unlike such an analysis of the data flow $X\to T$ \cite{information-flow}, the mutual information $I(W,Z;S)$ in the derived upper bound quantifies the information of the wireless channel contained in the model and directly reveals the robustness against wireless distortion. 
    Moreover, our measure is invariant to the choice of activation functions, since it is not directly influenced by hidden representations like $I(X;T)$.
\end{remark}

\begin{remark}
    In our previous analysis, we assumed that the loss function $l(w,z)$ is $\sigma$-sub-Gaussian, enabling us to leverage the properties of sub-Gaussian distributions in deriving the upper bound. However, this assumption may be too restrictive in certain scenarios where the loss function exhibits heavier tails. To address this limitation, we extend our theoretical results to the sub-Gamma distribution, which is a generalization of the sub-Gaussian distribution. Since both sub-Gaussian and sub-Gamma distributions are commonly used assumptions in deep learning theoretical analysis \cite{sub-gamma}, this extension will further enhance the applicability of the proposed theorem.

    A random variable $X$ is said to be sub-Gamma if for all $0<\lambda < c^{-1}$:
    \begin{equation}
        \log \E\left[e^{\lambda X}\right] \leq \frac{\lambda^2 \sigma^2}{2(1-c\lambda)},
    \end{equation}
    where $\sigma>0$ and $c>0$ are parameters characterizing the variance and tail behavior, respectively. Sub-Gamma distributions generalize sub-Gaussian distributions (where $c=0$) by allowing for heavier tails. Under this assumption, we can also derive an analogous upper bound on the performance degradation caused by imperfect wireless channels. 
    \begin{corollary}
        \label{thm:upper_bound_gamma}
        If the loss function $l(w,z)$ is $(\sigma,c)$-sub-Gamma on distribution $p(w,z)$, then the wireless risk discrepancy is upper bounded by
        \begin{equation}
            \begin{split}
                \Tilde{g}(\mu,f) & \triangleq \E_{\mu(s)}\left[ L(S) - L\right] \\
                & \leq \sigma \sqrt{2I(W,Z;S)} + cI(W,Z;S).
            \end{split}
        \end{equation}
    \end{corollary}
    The proof is provided in Appendix \ref{apx:proof_gamma}. The upper bound in the sub-Gamma scenario differs slightly from \cref{equ:upper_bound}. Specifically, we can only provide an upper bound on the discrepancy between wireless risk and standard risk rather than the absolute value in \cref{equ:wireless_discrepancy}. To accommodate the more general loss functions, the resulting bound is relatively looser compared to the sub-Gaussian assumption. Such an extension broadens the applicability of our theoretical analysis of the inherent robustness in wireless distributed learning.
\end{remark}


\subsection{Task-aware $\epsilon$-Capacity} 
The wireless risk discrepancy is a statistical measure averaged over the probability space, providing a general view. However, to delve deeper into the possible capacity gains benefiting from the robustness, we pay special attention to each detailed event.
Specifically, for an individual inference $(z,w,s)$, where $z$, $w$, and $s$ represent a particular sample, hypothesis, and channel condition, respectively, the difference between the actual wireless loss $\hat{l}(\hat{y}(x,w,s),y)$ and the true risk $L$ would possibly exceed the upper bound $G$ in Theorem \ref{thm:upper_bound}.
This deviation depends on the variability of the performance deterioration across different realizations, and it could occur due to varying samples, hypotheses, and channel conditions. 

To better analyze the probabilistic nature of wireless distributed learning, we introduce the concept of \textit{task outage} as in \cite{outage-scan} and \cite{outage-semantic-text}, which extends the outage in conventional communication systems towards the domain of wireless distributed learning. 
The conventional concept of outage is the event that occurs whenever the capacity falls below a target rate $R$ \cite{fundamental-wireless-communication}. 
However, given the robustness of learning models, lossy transmission is acceptable in wireless distributed learning as long as the performance is not severely affected. Therefore, we focus on the event that the instantaneous loss difference transcends the upper bound $G$ of the wireless risk discrepancy.
\begin{definition}
    \label{thm:outage}
    For an individual wireless distributed learning $(z,w,s)$ with a transmission rate $R$, a \textit{task outage} occurs when the following condition is satisfied. Then we define a random variable $E$ as the corresponding indicator
    \begin{equation}
        \label{equ:task_error}
        E=\mathbb{I}\left(\left| \hat{l}(\hat{y}_R(x,w,s),y) - L\right| \geq \sigma\sqrt{2I(W,Z;S)}\right),
    \end{equation}
    where $\mathbb{I}(\cdot)$ is the indicator function, $\hat{y}_R(x,w,s) = f_d^{(w)}(\hat{m})$ is the output of the decoder, $\hat{m}$ is the received signal that is transmitted through wireless channel $s$ at rate $R$. 
    The set of wireless losses that causes task outage can be written as $\mathcal{E}=\left\{\hat{l}: E = 1 \right\}$.
    Thus, the task outage probability is expressed as 
    \begin{equation}
        \label{equ:outage_probability}
        p_e\triangleq \Pr\{E=1\}.
    \end{equation}
\end{definition}

In the conventional communication system, \cite{general-capacity} defines $\epsilon$-achievable rate in the context of lossy transmission with error $\epsilon$. 
Given the task outage probability, we adopt this notion and apply it to the lossy transmission in wireless distributed learning. 

\begin{definition}
    \label{thm:task_achievable_rate}
    An $(n,K,\epsilon)$ code has block length $n$, $K$ codewords, and task outage probability not larger than $\epsilon$. $R\geq 0$ is a \textit{task-aware $\epsilon$-achievable rate} if, for every $\delta>0$, there exist, for all sufficiently large $n$, $(n, K, \epsilon)$ codes with rate
    \begin{equation}
        \frac{\log K}{n} > R - \delta.
    \end{equation}
\end{definition}

\begin{theorem}
    \label{thm:max_rate}
    If $R\geq 0$ is task-aware $\epsilon$-achievable, then
    \begin{equation}
        \label{equ:max_rate}
        R \leq \frac{1}{1-\epsilon}C\triangleq C_\epsilon,
    \end{equation}
    where $C=\liminf_{n\to\infty}\sup_{M^n}\frac{1}{n}I(M^n;\hat{M}^n)$ represents the Shannon's channel capacity. The maximum task-aware $\epsilon$-achievable rate is called the task-aware $\epsilon$-capacity $C_\epsilon$.
\end{theorem}

\begin{proof}
    A message $J$, drawn from the index set $\{1,2,\cdots, K\}$, is the output of the model encoder and results in the signal $M^n(J)$ by a channel encoder. The received signal $\hat{M}^n$ is used to guess the message by a channel decoder $\hat{J}=g(\hat{M}^n)$.
    \begin{subequations}
        \label{equ:thm_proof1}
        \begin{align}
            H(J|\hat{J}) & \leq H(J, E | \hat{J}) \label{equ:thm2_proof1.1} \\
            & = H(E|\hat{J}) + H(J|\hat{J}, E) \\
            & \leq H(E) + H(J|\hat{J}, E) \label{equ:thm2_proof1.2} \\
            & = H(E) + p_eH(J|\hat{J}, E=1) \\
            & \notag \quad + (1-p_e)H(J|\hat{J}, E=0) \\
            & \leq H(E) + p_eH(J|\hat{J}, E=1) \\
            & \leq H(E) + p_e\log(|\mathcal{J}|-1) \label{equ:thm2_proof1.3} \\
            & \leq H(E) + \epsilon \log K\ ,
        \end{align}
    \end{subequations}
    where \cref{equ:thm2_proof1.1} is from data processing inequality, \cref{equ:thm2_proof1.2} is due to the fact that conditioning can only decrease entropy, \cref{equ:thm2_proof1.3} is because $H(J|\hat{J},E=1)$ is maximized when $p(J|\hat{J},E=1)$ is uniformly distributed, which yields $\log(\left|\mathcal{J}\right|-1)$. 
    Hence, for any coding scheme,
    \begin{subequations}
        \begin{align}
            \log K & = H(J) \label{equ:thm2_proof2.1} \\
            & = H(J|\hat{J}) + I(J;\hat{J}) \\
            & \leq H(J|\hat{J}) + I(M^n;\hat{M}^n) \label{equ:thm2_proof2.2} \\
            & \leq H(E) + \epsilon\log K + I(M^n;\hat{M}^n) \label{equ:thm2_proof2.3} \\
            & \leq H(E) + \epsilon\log K + \sup_{M^n}I(M^n;\hat{M}^n) \label{equ:thm2_proof2.4},
        \end{align}
    \end{subequations}
    where \cref{equ:thm2_proof2.1} is because $J$ is uniformly distributed, \cref{equ:thm2_proof2.2} is from data processing inequality, and \cref{equ:thm2_proof2.3} is from the result obtained with \cref{equ:thm_proof1}. 
    Therefore, for every $\delta>0$
    \begin{equation}
        R-\delta < \frac{1}{1-\epsilon}\left[\frac{1}{n}\sup_{M^n}I(M^n;\hat{M}^n) + \frac{H(E)}{n}\right],
    \end{equation}
    which, in turn, implies
    \begin{equation}
        R\leq \frac{1}{1-\epsilon} \liminf_{n\to\infty}\frac{1}{n}\sup_{M^n}I(M^n;\hat{M}^n).
    \end{equation}
\end{proof}

\begin{remark}
    It can be seen that the task-aware $\epsilon$-capacity falls back to Shannon's channel capacity when $\epsilon = 0$. 
    In this scenario, the task outage would never occur, indicating that the wireless loss $\hat{l}$ must fall within a bounded interval $(L-G,L+G)$. 
    Since loss functions in typical applications are generally unbounded or half-bounded, this creates a contradiction. The only solution is a degenerate distribution where the actual loss is deterministic, i.e., $\text{Var}(\hat{l}) = 0$. Thus, it corresponds to the error-free transmission in data-oriented communication, as discussed in Remark \ref{thm:special_case}.
    On the other hand, the special case of $p_e=1$ is unattainable for common distributions in AI applications, preventing the task-aware $\epsilon$-capacity from reaching infinity. 
\end{remark}
\begin{remark}
    We can observe from \cref{equ:max_rate} that the task-aware $\epsilon$-capacity increases with $p_e$.
    By exploring the definition of task outage probability $p_e=\Pr\left\{\left|\hat{l}-L\right|\geq G\right\}$, we perceive $p_e$ as a two-tailed probability of the distribution of $\hat{l}$. Several major factors contribute to the augmentation of probability. 
    \begin{itemize}
        \item First, an increase in the variance of $\hat{l}$ or the transition to a heavier-tailed distribution may be noted. This indicates a high BER in communication, which arises from a high transmission rate. Therefore, with the increase of task-achievable rate, a higher capacity gain could be obtained.
        \item Second, a small cut-off value $G$ also contributes to the increase of the two-tailed probability. 
        As analyzed in Remark \ref{thm:chain_rule}, a small wireless risk discrepancy upper bound $G$ helps alleviate the over-fitting on a specific channel and thus enhances the model's robustness against the impact of wireless channels. 
        Therefore, a high task-aware $\epsilon$-capacity might be achieved despite the channel distortion. 
        \item Third, the standard risk $L$ moving closer towards the tails of the distribution may increase $p_e$. This indicates a larger deviation of $L$ from the mean of $\hat{l}$. As a result, the wireless risk discrepancy approaches the upper bound $G$, which implies the model's capability of fully exploiting the intrinsic robustness in wireless distributed learning. Thus, a more substantial capacity gain could be possible.
    \end{itemize}
\end{remark}

\section{Numerical Estimation of Upper Bound}
\label{sec:numerical_example}
By deriving an information-theoretic upper bound for the performance deterioration, we shed light on how to exploit the task-aware $\epsilon$-capacity from the robustness in wireless distributed learning.
The key challenge is the calculation of $I(W,Z;S)$ in a practical setting. In this section, we address this issue by utilizing an NN-based model $f^{(W)}$ in wireless distributed learning as an example.

In the context of NNs, the hypothesis $W$ represents the weights of the NN, and the sample space $\mathcal{Z}$ represents the dataset.
As specified in Remark \ref{thm:chain_rule}, we turn to estimating the conditional mutual information $I(W;S|Z)$ for convenience. It can be written as the expectation of the KL divergence between $p(w|s,z)$ and $p(w|z)$ as, 
\begin{equation}
    \label{equ:KL_divergence_conditional}
    I(W ; S | Z) = \E_{p(s,z)}\left[D_{\text{KL}}(p(w | s, z) \| p(w | z))\right].
\end{equation}
Considering the Bayes equation $p(w|s,z)=\frac{p(w|z)p(s|w,z)}{p(s|z)}$, $p(w|z)$ can be interpreted as prior knowledge of the weights depending on a given dataset $z$. Meanwhile, $p(w|s,z)$ represents the posterior distribution of the weights given a dataset $z$ and the evidence of side information $s$. Prior $p(w|z)$ is the probability of the NN's weights before learning about the wireless channel, i.e., weights in standard distributed learning.
Following the common assumption in deep learning theory, we assume both $p(w|z)=\mathcal{N}(w|\theta_z,\Sigma_z)$ and $p(w|s,z)=\mathcal{N}(w|\theta_{s,z},\Sigma_{s,z})$ are Gaussian. The means $\theta_{z}, \theta_{s,z}$ are the yielded hypotheses of the learning algorithm given the corresponding knowledge. Thus, the KL divergence term has a closed-form expression as
\begin{equation}
    \label{equ:KL_divergence_Gaussian}
    \begin{split}
        D_{\text{KL}}(p(w | s, z) \| p(w | z)) = \frac{1}{2}\left[\log \frac{\det \left(\Sigma_{s,z}\right)}{\det \left(\Sigma_z\right)}-D\right.\\
        \left.+\left(\theta_{s,z}-\theta_z\right)^{\top} \Sigma_z^{-1}\left(\theta_{s,z}-\theta_z\right)+\operatorname{tr}\left(\Sigma_z^{-1} \Sigma_{s,z}\right)\right],
    \end{split}
\end{equation}
where $\det$ and $\operatorname{tr}$ represent the determinant and trace of a matrix, respectively, while $D$ represents the dimension of weight $w$, which is a constant for a specific model $f^{(w)}$. For simplicity, we further assume the covariances of both the prior and the posterior are proportional, which is a common practice in PAC-Bayes analysis \cite{proportional-covariance}. As a result, the logarithmic and trace terms in \cref{equ:KL_divergence_Gaussian} become constant and the conditional mutual information can be expressed as
\begin{equation}
    \label{equ:conditional_MI}
    I(W ; S | Z) \propto \E_{p(s,z)}\left[\left(\theta_{s,z}-\theta_z\right)^{\top} \Sigma_z^{-1}\left(\theta_{s,z}-\theta_z\right)\right].
\end{equation}

The prior covariance can be approximated using bootstrapping \cite{PAC-Bayes} by
\begin{equation}
    \label{equ:bootstrapping}
    \Sigma_z \approx \frac{1}{K} \sum_{k=1}^K\left(\theta_{(s,z)_k}-\theta_{s,z}\right)\left(\theta_{(s,z)_k}-\theta_{s,z}\right)^{\top},
\end{equation}
where $(s,z)_k$ is a bootstrap sampling set from space $\mathcal{S}\times\mathcal{Z}$ in the $k$-th experiment. During the training process of deep learning, it is prohibitive to learn the weight $\theta_{(s,z)_k}$ from $K$ bootstrapping samples. Therefore, we further utilize the influence function (IF) from robust statistics \cite{influence-function} to approximate the difference $\theta_{(s,z)_k}-\theta_{s,z}$. The following lemma adapted from \cite{IF} formalizes the approximation.

\begin{lemma}[Influence Function]
    \label{thm:IF}
    Assume in Poisson bootstrapping that the sampling weights, denoted by ${\xi}=(\xi_1,\cdots,\xi_n)^\top$, follow a binomial distribution $\xi_i\sim \text{Binomial}(n,\frac{1}{n})$, where $\lim_{n\to\infty} \text{Binomial}(n,\frac{1}{n})\to\text{Poisson}(1)$. 
    Given new hypotheses $\hat{\theta}_{(s,z),{\xi}}\triangleq \arg\min \frac{1}{n}\sum_{i=1}^n \xi_i \hat{l}_i(\theta)$ \footnote{To simplify the notation, we denote the wireless loss based on the weight $\theta$ as $\hat{l}_i(\theta) = l(\hat{y}, y)$, where $\hat{y}$ is the target result of model $f^{(\theta)}$.} and $\hat{\theta}_{s,z}\triangleq \arg\min \frac{1}{n}\sum_{i=1}^n \hat{l}_i(\theta)$, the approximation of the difference is defined by
    \begin{equation}
        \label{equ:IF}
        \hat{\theta}_{(s,z),{\xi}} - \hat{\theta}_{s,z} \approx \frac{1}{n}\sum_{i=1}^n (\xi_i-1)\psi_i = \frac{1}{n}{\psi}^\top({\xi}-{1}),
    \end{equation}
    where $\psi_i = -{H}_{\hat{\theta}_{s,z}}^{-1}\nabla_\theta \hat{l}_i(\hat{\theta}_{s,z})$ represents the IF, ${\psi} = (\psi_1,\cdots,\psi_n)^\top$ represents the collection of IFs, and ${H}_{\hat{\theta}_{s,z}} = \frac{1}{n}\sum_{i=1}^n\nabla_\theta^2\hat{l}_i(\hat{\theta}_{s,z})$ represents the Hessian matrix. As a result, the covariance can be approximated as
    \begin{equation}
        \label{equ:fisher_approx}
        \Sigma_z \approx \frac{1}{K} \sum_{k=1}^K\left(\hat{\theta}_{(s,z),{\xi}_k}-\hat{\theta}_{s,z}\right)\left(\hat{\theta}_{(s,z),{\xi}_k}-\hat{\theta}_{s,z}\right)^{\top} \approx \frac{1}{n}{F}_{\hat{\theta}_{s,z}}^{-1},
    \end{equation}
    where ${\xi}_k$ represents the sampling weights in the $k$-th experiment and ${F}_{\hat{\theta}_{s,z}}$ is the Fisher information matrix.
\end{lemma}

Now we are able to approximate the conditional mutual information in \cref{equ:conditional_MI} as 
\begin{equation}
    \label{equ:conditional_MI_approx}
    \begin{split}
        I(W;S|Z) & \propto n \E_{p(s,z)}\left[\left(\theta_{s,z}-\theta_z\right)^{\top} {F}_{\hat{\theta}_{s,z}}^{-1}\left(\theta_{s,z}-\theta_z\right)\right] \\
        & \approx n \left(\overline{\theta}_{s,z}-\theta_z\right)^{\top} {F}_{\hat{\theta}_{s,z}}^{-1}\left(\overline{\theta}_{s,z}-\theta_z\right),
    \end{split}
\end{equation}
where we employ quadratic mean $\overline{\theta}_{s,z}=\sqrt{\frac{1}{K}\sum_{k=1}^K\hat{\theta}_{(s,z),{\xi}_k}^2}$ to better estimate $\theta_{s,z}$ in the expectation. In practice, we could pre-train a model in standard distributed learning to obtain $\theta_{z}$, and we denote $\Delta\theta = \overline{\theta}_{s,z}-\theta_z$. By expanding the Fisher information matrix, \cref{equ:conditional_MI_approx} can be simplified as
\begin{equation}
    \label{equ:conditional_MI_efficient}
    \begin{split}
        I(W;S|Z) & \propto n \Delta\theta^{\top} \left[\frac{1}{T}\sum_{t=1}^T\nabla_\theta \hat{l}_t(\hat{\theta}_{s,z})\nabla_\theta \hat{l}_t(\hat{\theta}_{s,z})^\top\right]^{-1}\Delta\theta \\
        & = \frac{n}{T}\sum_{t=1}^T\left[\Delta\theta^\top\nabla_\theta \hat{l}_t(\hat{\theta}_{s,z}) \right]^2\triangleq \Tilde{I}(W;S|Z),
    \end{split}
\end{equation}
where $n$ denotes the number of samples and $T$ denotes the number of iterations to estimate Fisher information. In this way, the complexity of estimating the conditional mutual information becomes $O(DT)$, and the corresponding algorithm is summarized in Algorithm \ref{alg:approximation}.

\begin{algorithm}
    \fontsize{9pt}{9pt}\selectfont
    \caption{Efficient approximation of conditional mutual information $I(W;S|Z)$}
    \label{alg:approximation}
    \begin{algorithmic}[1]
        \REQUIRE Number of samples $n$, number of batches $B$, number of estimation iterations $T$, average parameters $(\rho,K)$.
        \ENSURE Approximated $\Tilde{I}(W;S|Z)$.
        \STATE Pre-train in standard distributed learning, $\theta_z\gets p(w|z)$.
        \FOR{$t=1:B$}
            \STATE Train in wireless distributed learning, calculate and store the gradient $\nabla l_t$.
            \STATE Back-propagation and update weight $\hat{\theta}_t$.
            \STATE Apply moving average hypotheses over $K$ iterations $\overline{\theta}_t\gets \sqrt{\rho \overline{\theta}_{t-1}^2 + \frac{1-\rho}{K}\sum_{k=0}^K\hat{\theta}_{t-k}^2}$.
        \ENDFOR
        \STATE $\Delta\theta\gets \overline{\theta}_B - \theta_z$, $\Delta F\gets 0$.
        \FOR{$t=1:T$}
            \STATE $\Delta F_t\gets \Delta F_{t-1} + (\Delta\theta ^\top \nabla l_t)^2$.
        \ENDFOR
        \STATE $\Tilde{I}(W;S|Z)\gets \frac{n}{T}\Delta F_T$.
    \end{algorithmic}
\end{algorithm}

To validate the effectiveness of the proposed task-aware $\epsilon$-capacity in Section \ref{sec:main_results}, we utilize a classic binary classification dataset of ASIRRA Cats \& Dogs \cite{dogs-vs-cats} and employ a 6-layer convolutional neural network (CNN) for wireless distributed learning system.
Specifically, the encoder $f_e^{(W)}$ at the UE is comprised of the first 3 layers and the decoder $f_d^{(W)}$ at the BS is comprised of the remaining part of the model. The parameter $K$ in Algorithm \ref{alg:approximation} is both the number of bootstrapping samples and the number of iterations for the moving average. This parameter determines both the estimation accuracy of $\overline{\theta}_{s,z}$ and the computational complexity. In our experiment, we set $K=5$ since it achieves similar accuracy as iterating over the entire training dataset $K=n$. Another parameter $\rho$ in Algorithm \ref{alg:approximation} is set to 0.99, as the typical setting in the RMSProp and Adam algorithms. The number of batches and the number of estimation iterations are set to $B=128$ and $T=5$, respectively.
Regarding the wireless channel setting, we use the additive white Gaussian noise (AWGN) and Rayleigh channels. 
Moreover, we adopt the assumption of a quasi-static fading model that the channel gain changes slowly and remains static during the transmission of a block.
The channel coding and modulation in the physical layer comply with the 3GPP specification \cite{3GPP-38.214} to ensure compatibility with modern wireless communication systems. In particular, the low-density parity-check (LDPC) code is used for channel coding, and modulations of QPSK, 16QAM, 64QAM, and 256QAM are supported. 

First, we train the model under both standard and wireless distributed learning. We adopt the differentiable quantization technique to enable training with LDPC codes and modulation.
Specifically, in the back-propagation of the training process, the derivative of the quantization function is zero everywhere except at integers, where it is undefined. Therefore, to enable transmission techniques like coding and modulation, we must replace its derivative in the back-propagation. The study in \cite{differentiable-quantize} summarizes various differentiable quantization methods that empirically work well, such as the soft-max weighted quantization in \cite{soft-quantize} and replacing quantization with additive uniform noise in \cite{uniform-noise-quantize}.
The empirical risks $\hat{L}$ and $\hat{L}(s)$ are averaged over multiple training results. Furthermore, we calculate the average difference across multiple simulations to derive $\hat{g}(\mu,f)=\frac{1}{n}\sum_i|\hat{L}(s_i)-\hat{L}|$. 
We approximate the upper bound $\hat{G}$ via Algorithm \ref{alg:approximation} and statistically compute the task outage probability $p_e$.
The outcomes are summarized in Table \ref{tab:comparison}, with the parentheses denoting the power of channel noise, modulation, and code rate.
From the table, the actual wireless risk discrepancy is maintained within the theoretical upper bound across various channel conditions. Hence, they provide a practical example of the proposed theorems on the wireless risk discrepancy and wireless distributed learning system's inherent robustness. 
However, we acknowledge that the proposed upper bound can be quite loose under certain conditions. The tightness of the bound has not yet been theoretically evaluated in practical systems. We expect in our future exploration that applying more constraints on the loss function or utilizing more refined inequalities may yield tighter bounds.
Additionally, the validation accuracy and the BER of wireless transmission are summarized in the last two columns. It can be seen that wireless distributed learning still performs well under a large BER, demonstrating its robustness against wireless channel noise.

\begin{table*}[htbp]
  \caption{Comparison between loss difference and upper bound}
  \label{tab:comparison}
  \centering
  \begin{tabular}{| c | c | c | c | c | c |}
    \hline
    Channel & $\hat{g}(\mu,f)$ & $\hat{G}$ & $p_e$ & Acc & BER \\
    \hline
    AWGN (5dB,16QAM,R=1/2)      & 0.0350 & 0.0518 & 0.2881 & 0.8152 & 0.1263 \\ \hline 
    AWGN (10dB,64QAM,R=1/2)     & 0.0124 & 0.0423 & 0.2697 & 0.8344 & 0.0670 \\ \hline 
    AWGN (15dB,64QAM,R=3/4)     & 0.0034 & 0.0390 & 0.2620 & 0.8318 & 0.0406 \\  \hline 
    Rayleigh (5dB,16QAM,R=3/8)  & 0.0443 & 0.0769 & 0.2856 & 0.8196 & 0.1458 \\ \hline 
    Rayleigh (10dB,64QAM,R=5/12) & 0.0270 & 0.0656 & 0.2746 & 0.8284 & 0.1076 \\ \hline 
    Rayleigh (15dB,64QAM,R=2/3) & 0.0222 & 0.0505 & 0.2721 & 0.8348 & 0.0944 \\ \hline 
  \end{tabular}
\end{table*}

Additionally, we illustrate the task-aware $\epsilon$-capacity resulting from the robustness in wireless distributed learning. 
Specifically, we perform the distributed inference task with a trained model in a wireless system at various transmission rates, achieved by employing different code rates and modulation schemes. To be compatible with the bit-level communication system, the output features of the encoder are quantized into binary sequences during the inference process. In particular, the floating-point precision elements are quantized to 1-bit fixed-point binary representations. 
In Figure \ref{fig:loss_vs_rate}, we plot the performance curves with loss on the x-axis and physical layer channel transmission rate on the y-axis for a clear presentation. 
The transmission rate $R$ is defined as the number of bits transmitted over the wireless channel in the physical layer. It is computed by $R=\log_2(Q)\cdot R_c$, where $Q$ is the modulation order and $R_c$ is the channel code rate.
The Shannon channel capacity $C$ and the task-aware $\epsilon$-capacity $C_\epsilon$ in \cref{equ:max_rate} are also noted in the figures. According to Definition \ref{thm:outage}, we identify the complement of task outage set $\mathcal{E}^c$ as the \textit{task-achievable region}. 
In the simulation, the loss is averaged over the dataset. It is valid to assert that the task outage probability remains below $p_e$ when $\left|\hat{l}-L\right|<G$. 
The task-achievable region is shaded in the figures.

It can be seen from the figures that there is an increase in loss within the task-achievable region. It is worth noting that the modern channel coding schemes we used operate under the constraint of finite code block lengths. Even when the transmission rate is below the channel capacity, these codes cannot guarantee a zero BER. As the transmission rate approaches the channel capacity, the BER increases due to the limitations of practical coding and decoding algorithms. Therefore, the distorted received signals could slightly affect the model's inference process, leading to an increase in the loss function. However, this increase does not necessarily indicate a significant degradation in classification accuracy. 
For instance, consider this binary classification task: if the model's output probabilities shift from (0.9, 0.1) to (0.6, 0.4), the loss increases due to reduced confidence, but the classification result remains the same.
Based on our study, when the actual wireless loss $\hat{l}$ does not exceed the upper bound, we can still consider it a robust inference under wireless channel noise.

It can be observed from Figure \ref{fig:loss_vs_rate} that there exists a transmission rate $R$ with which the performance falls at the boundary of the task-achievable region. Therefore, a successful task can be accomplished by employing such a transmission rate. Notably, the rate $R$ surpasses the Shannon channel capacity $C$, thus corroborating the intrinsic robustness in wireless distributed learning. Furthermore, we compute the task-aware $\epsilon$-capacity $C_\epsilon$ in \cref{equ:max_rate} using the task outage probability $p_e$ from Table \ref{tab:comparison} as $\epsilon$. It can be seen from Figure \ref{fig:loss_vs_rate} that the achievable rate $R$ falls below $C_\epsilon$, aligning with the proposed theorem.

\begin{figure}[htbp]
  \centerline{\subfigure[AWGN channel.]{\includegraphics[width=0.5\columnwidth]{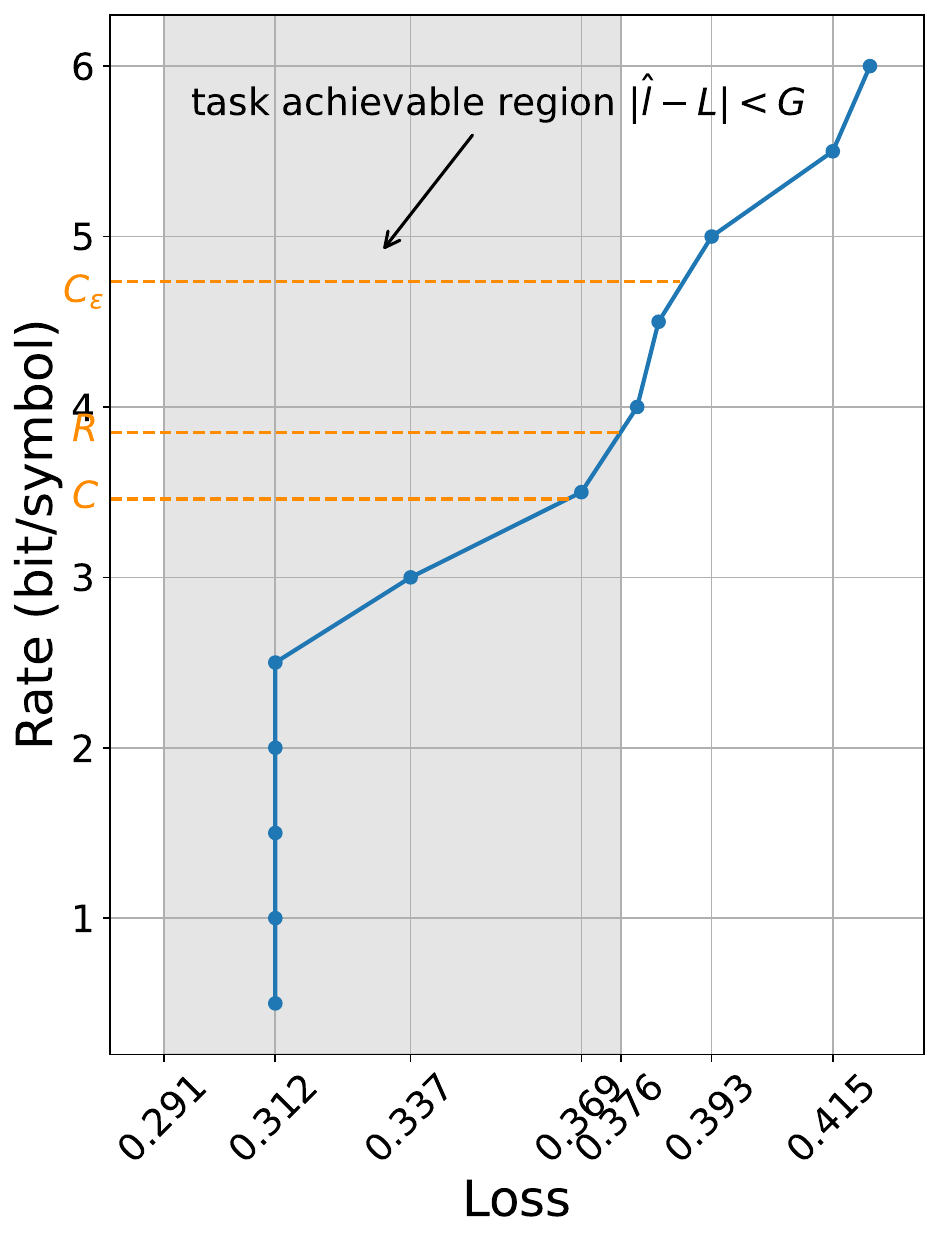}
      \label{fig:first_case}}
    \subfigure[Rayleigh channel.]{\includegraphics[width=0.5\columnwidth]{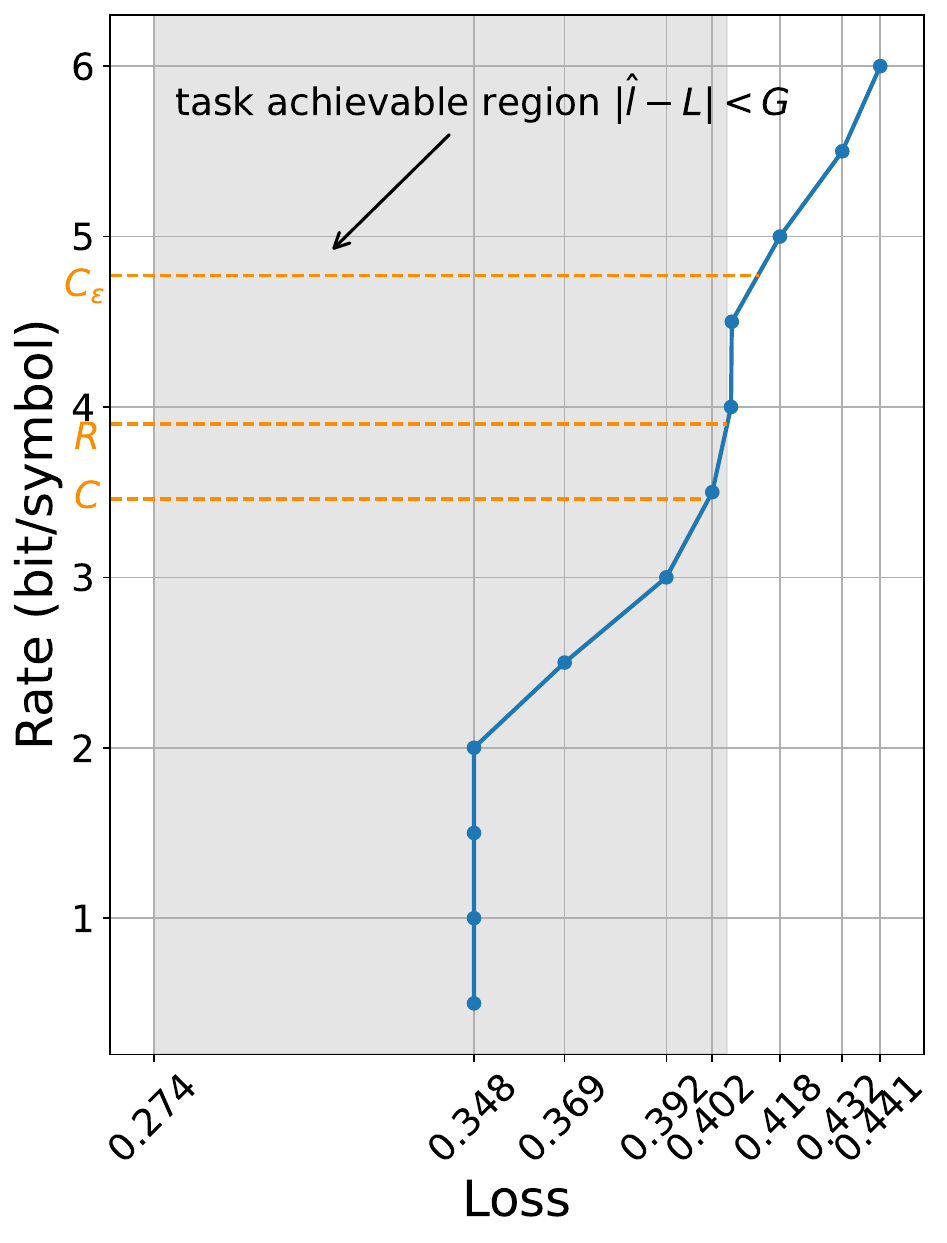}
      \label{fig:second_case}}}
  \caption{The performance curves at different rates in AWGN and Rayleigh channels. The shadows denote the task-achievable region.}
  \label{fig:loss_vs_rate}
\end{figure}

\section{Proposed robust training framework}
\label{sec:robust_training_framework}
As analyzed in Section \ref{sec:system_model}, a small wireless risk discrepancy indicates that the model is more robust against imperfect wireless channels. However, since the discrepancy is defined on expected statistical results, it is relatively difficult to estimate in advance. Therefore, the upper bound of the wireless risk discrepancy derived in Section \ref{sec:upper_bound} can serve as an alternative to measuring the robustness of wireless distributed learning. 
By minimizing the upper bound $G=\sigma\sqrt{2I(W,Z;S)}$, the model could generalize to various imperfect wireless channels while maintaining inference accuracy. 
For this purpose, we propose a robust training framework that optimizes the distributed feature decoder within an NN-based model framework.

\subsection{Accuracy and Robustness Trade-off}

The mutual information $I(W,Z;S)$ is determined by the model weights, dataset, and the channel side information. For a specific task in wireless distributed learning, the dataset, and the channel state are often given or fixed. 
Therefore, it is natural to minimize the mutual information in the training process by constructing an objective function as follows:
\begin{equation}
    \label{equ:pib_objective1}
    \min_{p(w|s,z)} \E_{p(w|s,z)} \E_{p(s,z)}\left[l(W,Z)\right] + \beta I(W,Z;S).
\end{equation}
The first term represents the wireless risk associated with the task performance of wireless distributed learning. The second one corresponds to the mutual information in the derived upper bound, which aims to minimize the performance deterioration due to the wireless channel noise.
The parameter $\beta>0$ is the weight of the regularization term.

For classification tasks, the loss function is the cross-entropy between the distributions of wireless prediction $p(y|x,w,s)$ and the target $p(y|x)$. As shown in \cite{CE-MI}, minimizing the cross-entropy is equivalent to maximizing the mutual information. Hence, the objective function in \cref{equ:pib_objective1} is equivalent to 
\begin{equation}
    \label{equ:pib_objective2}
    \max_{p(w|s,z)} I(Y;W,S|X) - \beta I(W,Z;S).
\end{equation}
Similar to the IB framework, the objective function demonstrates a trade-off between the sufficiency of information on the task target label $Y$ and the minimality of the learned weight $W$.
This trade-off aims for learned weights that have minimal information about the channel state information to avoid over-fitting while having sufficient information about the labels \cite{minimal-sufficient}. Specifically, the set of sufficient weights $\mathcal{S}$ and minimal sufficient weights $\mathcal{M}$ are defined as follows:
\begin{equation}
    \label{equ:minimal-sufficient}
    \begin{split}
        \mathcal{S} & \coloneqq \arg\max_{W'} I(Y;W',S|X),\\
        \mathcal{M} & \coloneqq\arg \min_{W'\in\mathcal{S}}I(W',Z;S).
    \end{split}
\end{equation}
Based on the chain rule of conditional mutual information, the first term could be expanded as
\begin{equation}
    I(Y;W,S|X) = I(Y;W|X) + I(Y;S|W,X).
\end{equation}
Combined with \cref{equ:MI_second}, we have
\begin{equation}
    \label{equ:MI_tradeoff}
    I(Y;W,S|X) - I(Y;W|X) = I(W,Z;S) - I(W,X;S).
\end{equation}
The second term on the right-hand side could be further expressed as $I(W,X;S) = I(X;S) + I(W;S|X)=I(W;S|X)$, where $I(X;S)=0$ since $Z$ and $S$ are independent. It can be seen from \cref{equ:MI_tradeoff} that the sufficiency $I(Y;W,S|X)$ increases with minimality $I(W,Z;S)$, thereby establishing the trade-off between them.

The objective function is a weighted sum of two mutual information terms with a parameter controlling the trade-off. Specifically, task accuracy is characterized by the amount of information about the task in the model weights, quantified by $I(Y;W,S|X)$. Intuitively, if the model learns more about the target and the wireless channel conditions, the BS could recover a high-quality feature leading to higher accuracy. However, as the minimal sufficient statistic in the IB framework, learning too much about the dataset, measured by $I(W,Z;S)$, may result in over-fitting. Thus, there is an inherent trade-off between the inference performance and the robustness against channel distortions. The objective function in \cref{equ:pib_objective2} explicitly formulates such an accuracy-robustness trade-off for wireless distributed learning system.

\subsection{Optimal Solution}
Given the dataset $Z$ and channel side information $S$, our goal is to find the optimal weights $p(w|s,z)$ that minimize the objective function. The variational method is a common way to solve the mutual information trade-off in \cref{equ:pib_objective1}. First, we consider the conditional distribution constraint of $p(w|s,z)$. The problem could be formulated as
\begin{align}
    \begin{split}
        \min_{p(w|s,z)} & \E_{p(w|s,z)} \E_{p(s,z)}\left[l(W,Z)\right] + \beta I(W,Z;S),\\
        & \text{s.t.}\quad \E \left[p(w|s,z)\right] =1.
    \end{split}
\end{align}
Using the Lagrange multiplier method, we have the following lemma for the optimal solution:
\begin{lemma}
\label{thm:optimal_solution}
The optimal posterior $p(w|s,z)$ that minimizes \cref{equ:pib_objective1} satisfies the equation
\begin{equation}
    \label{equ:optimal_solution}
    p(w|s,z) = \frac{p(w|z)}{A(s,z)}\exp\left(-\frac{1}{\beta}\hat{L}_{s,z}(w)\right),
\end{equation}
where $A(s,z)$ is a normalizing constant.
\end{lemma}

\begin{proof}
    By introducing the Lagrange multipliers $\lambda(s,z)$ for the normalization of the conditional distribution, the Lagrangian function could be written as
    \begin{align}
        \label{equ:lagrange}
        \begin{split}
            \mathcal{L}= & \E_{p(w|s,z)} \left[\hat{L}_{s,z}(w)\right] + \beta \E_{p(w,s,z)}\left[\frac{\log p(w|s,z)}{p(w|z)}\right]\\ 
            & + \E_{p(s,z)}\lambda(s,z)\left(\E \left[p(w|s,z)\right] - 1\right),
        \end{split}
    \end{align}
    where the empirical wireless risk given a certain dataset and channel condition is denoted as $\hat{L}_{s,z}(w) = \frac{1}{n} \sum_{i=1}^n l(\hat{y}_i,y_i)$. 
    Additionally, the second term in \cref{equ:lagrange} follows from $I(W,Z;S)=I(W;S|Z)$ in Remark \ref{thm:chain_rule}.
    Taking derivatives with respect to the posterior $p(w|s,z)$ for given $s$ and $z$ results in
    \begin{align}
        \begin{split}
            \frac{\partial\mathcal{L}}{\partial p(w|s,z)} = \hat{L}_{s,z}(w) + \beta \left(\frac{\log p(w|s,z)}{p(w|z)}-1\right) + \lambda(s,z).
        \end{split}
    \end{align}
    Let $\frac{\partial\mathcal{L}}{\partial p(w|s,z)}=0$ and the variational result could be obtained as follows:
    \begin{align}
        \begin{split}
            \log p(w|s,z) & = -\frac{1}{\beta}\hat{L}_{s,z}(w) + \log p(w|z) - \Tilde{\lambda}(s,z), \\
            p(w|s,z) & = p(w|z) \exp\left(-\frac{1}{\beta}\hat{L}_{s,z}(w)\right)\exp\left(-\Tilde{\lambda}(s,z)\right),
        \end{split}
    \end{align}
    where $\Tilde{\lambda}(s,z) = \frac{\lambda(s,z)}{\beta} + 1$. The second exponential term $\exp\left(-\Tilde{\lambda}(s,z)\right)$ could be considered as the partition function that normalizes the conditional distribution. Let the normalized function $A(s,z)$ be defined as 
    \begin{equation}
        \begin{split}
            A(s,z) & = \exp\left(-\Tilde{\lambda}(s,z)\right) \\
            & = \E\left[p(w|z) \exp\left(-\frac{1}{\beta}\hat{L}_{s,z}(w)\right)\right],
        \end{split}
    \end{equation}
    thereby completing the proof. 
\end{proof}

To obtain the optimal posterior in \cref{equ:optimal_solution}, we further express it in the exponential form
\begin{equation}
    p(w|s,z) = \frac{1}{A(s,z)}\exp\left(\frac{1}{\beta} U(w)\right),
\end{equation}
where $U(w) = \hat{L}_{s,z}(w) - \beta \log p(w|z)$. 
This reveals that $p(w|s,z)$ follows a typical Gibbs distribution $\pi(w)\propto \exp\left(-\frac{1}{\beta} U(w)\right)$ with an energy function $U(w)$ and a temperature parameter $\beta$ corresponding to the regularization parameter in \cref{equ:pib_objective1}. 
According to \cite{Gibbs}, if $U$ has a unique global minimum, is sufficiently smooth, and properly behaved at the boundaries, the Langevin diffusion $X(t)$ in the following stochastic differential equation converges to the stationary distribution $\pi(x)\propto \exp\left(-\frac{1}{\beta}U(x)\right)$
\begin{equation}
    \label{equ:SDE}
    dX(t) = -\nabla U(X(t))dt + \sqrt{2\beta}dB(t),\quad t\geq 0,
\end{equation}
where $B(t)$ is the standard Brownian motion in $\R^D$. Specifically, the Gibbs distribution is the unique invariant distribution of \cref{equ:SDE}, and the distribution of $X(t)$ converges rapidly to $\pi(x)$ as $t\to\infty$. 
Furthermore, for sufficiently small values of the temperature $\beta$, the Gibbs distribution concentrates around the minimum of $U$. Thus, with high probability, a sample from the Gibbs distribution is an almost-minimizer of the energy function $U(X(t))$ \cite{almost-minimizer}.

Applying the Euler-Maruyama discretization \cite{discretization} to \cref{equ:SDE} yields the discrete-time Markov process:
\begin{equation}
    \label{equ:SGLD}
    X_{k+1} = X_k - \eta_k g_k + \epsilon_k\sqrt{2\eta_k \beta} ,
\end{equation}
where $g_k$ is a conditionally unbiased estimator of the energy function gradient $\nabla U(X_k)$, $\epsilon_k\sim \mathcal{N}(\mathbf{0},\mathbf{I}_D)$ is a standard Gaussian noise vector, and $\eta_k>0$ is the learning rate. This recursion corresponds to the weight update process of Stochastic Gradient Langevin Dynamics (SGLD) \cite{SGLD}. It has been demonstrated in \cite{Gibbs-convergence1} and \cite{Gibbs-convergence2} that the discrete SGLD recursion accurately tracks the Langevin diffusion. 
Under the conditions that $\sum_t^\infty \eta_t\to \infty$, $\sum_t^\infty \eta_t^2\to0$, and an annealing temperature parameter $\beta$, the distribution of $X_k$ will approximate the Gibbs distribution for sufficiently large $k$. Therefore, for large enough $k$, the output of the SGLD algorithm is also an almost-minimizer of the energy function $U(X_k)$.

Consequently, it is natural to adopt the SGLD algorithm to approximate the optimal posterior $p(w|s,z)$ in \cref{equ:optimal_solution} by setting the model weights $w$ as the Langevin diffusion $X_k$. 
The SGLD algorithm is capable of efficiently solving large-scale posterior inference problems. It achieves an exponential rate of convergence for a non-convex objective function \cite{Gibbs-convergence1}.
Using this algorithm, we can iteratively calculate the optimal weights through the following process:
\begin{equation}
    w_{k+1}=w_k - \eta_k\nabla U(w_k) + \epsilon_k\sqrt{2\eta_k\beta}.
\end{equation}
To compute the gradient of the energy function $U(w) = \hat{L}_{s,z}(w) - \beta \log p(w|z)$, we adopt the assumption from Section \ref{sec:numerical_example} that $p(w|z)=\mathcal{N}(w|\theta_z,\Sigma_z)$ is a Gaussian distribution in the deep learning scenario. 
Therefore, we have
\begin{equation}
    -\log p(w|z)\propto (w-\theta_z)^\top\Sigma_z^{-1}(w-\theta_z) + \log \det \Sigma_z,
\end{equation}
where the determinant term is equivalent to the sum of the logarithms of the eigenvalues $\log \det \Sigma_z = \sum_{i=1}^D\lambda_i$. 
Under this assumption, the gradient of the energy function is given by
\begin{equation}
    \nabla U(w) = \frac{\partial \hat{L}_{s,z}(w)}{\partial w} - 2\beta \Sigma_z^{-1}w,
\end{equation}
where the first term is the gradient of the loss function, which can be calculated using back-propagation.
The detailed SGLD training process for the proposed objective function \cref{equ:pib_objective1} is summarized in Algorithm \ref{alg:SGLD}. 
\begin{algorithm}
    \fontsize{9pt}{9pt}\selectfont
    \caption{SGLD training process}
    \label{alg:SGLD}
    \begin{algorithmic}[1]
        \REQUIRE Number of samples $n$, number of batches $B$, learning rate $\eta$, temperature $\beta$, decay scheme $\phi_\eta,\phi_\beta$.
        \ENSURE Model weights $\{w_k\}$ satisfying $p(w|s,z)$.
        \STATE Pre-train in standard distributed learning, $\theta_z\gets p(w|z)$.
        \REPEAT
            \STATE Calculate the mini-batch gradient of energy function $g_{k-1} \gets \nabla \left( -\frac{B}{n} \sum_b \log p(y_b | x_b, w_{k-1}, s_b)\right.$
            $\left.- 2\beta_{k-1} \log p(w_{k-1}|z) \right)$.
            \STATE Sample Gaussian random noise $\epsilon_k\gets \mathcal{N}(\mathbf{0},\mathbf{I}_D)$.
            \STATE Weights update $w_k\gets w_{k-1} - \eta_{k-1} g_k + \epsilon_k \sqrt{2\eta_{k-1} \beta_{k-1}} $.
            \STATE Learning rate and temperature decay $\eta_k\gets \phi_\eta(\eta_{k-1}),\beta_k\gets \phi_\beta(\beta_{k-1})$.
            \STATE $k\gets k+1$.
        \UNTIL{The weights $\{w_k\}$ converge.}
    \end{algorithmic}
\end{algorithm}

\section{Simulations}
\label{sec:simulation}
In this section, we evaluate the performance of the proposed robust training framework in Section \ref{sec:robust_training_framework} by comparing it with the vanilla and IB-based training mechanisms.

\subsection{Simulation Setup}
\begin{enumerate}
    \item \textbf{Datasets}: In this section, we select four benchmark datasets for image classification, including ASIRRA Cats \& Dogs, CIFAR-10, CIFAR-100 \cite{CIFAR}, and Tiny-ImageNet \cite{tinyimagenet}. The CIFAR-10 dataset consists of 60,000 color images in 10 classes with 6,000 images per class, including 50,000 training images and 10,000 test images. The CIFAR-100 dataset consists of the same number of color images as CIFAR-10, but it has 100 classes with 600 images per class. There are 500 training images and 100 testing images per class. The 100 classes in the CIFAR-100 are grouped into 20 superclasses. Each image comes with a "fine" label (the exact class) and a "coarse" label (the superclass). 
    The Tiny-ImageNet dataset is a subset of the ILSVRC-2012 classification dataset. It consists of 200 classes, and for each class it provides 500 training images and 50 validation images. The images have been downsampled to 64$\times$64$\times$3 pixels. For the image classification tasks, we adopt classification accuracy as the performance metric. 
    \item \textbf{NN architectures}: To demonstrate the effectiveness of the proposed robust training framework, we implement a wireless distributed learning system with two classic CNNs, the VGG network \cite{VGG} and the Residual Network (ResNet) \cite{ResNet}, and a Vision Transformer (ViT) \cite{vit}. Specifically, VGG11 and ResNet18 models are utilized for the CIFAR-10 classification task, VGG16 and ResNet34 models are employed for the CIFAR-100 classification task, and ViT is used for the Tiny-ImageNet classification task. As it is challenging to train ViT from scratch, we adopt the pre-trained version, ViT-small-patch16-224. 
    For the Cats \& Dogs dataset, the 6-layer CNN in Section \ref{sec:numerical_example} is adopted. The detailed architectures for the encoder $f_e^{(W)}$ and decoder $f_d^{(W)}$ are presented in Appendix \ref{sec:architecture}.
    \item \textbf{Baselines}: We compare the proposed robust training framework in wireless distributed learning against the vanilla training technique and the mutual information method in \cite{IB}, denoted as `Vanilla' and `MI', respectively. In particular, the vanilla method trains NNs with plain cross-entropy loss and Adam optimizer \cite{Adam}, which is a well-recognized stochastic optimization method. 
    The MI method utilizes an approximate KL divergence between the distributions of target variable $\hat{Y}$ and received signal $\hat{M}$ as the regularization term. Such a regularization is based on the IB principle, aiming at reducing the redundancy in the features while maximizing the task-relevant information.
    \item \textbf{Hyperparameters}: For the parameters in Algorithm \ref{alg:approximation}, we use the same setting as the example in Section \ref{sec:numerical_example}, where $K=5$, $\rho=0.99$, $B=128$, and $T=5$. In the SGLD training process, we use varied parameters $\beta$ and learning rate $\eta$ for different models and datasets. A high learning rate would lead to divergence during the training process, and a small learning rate would result in a slow convergence. Meanwhile, a high regularization parameter $\beta$ makes the model fail to learn information from the task, while a small $\beta$ might not sufficiently regularize the model. Therefore, we tune these parameters in our experiments for a stable convergence. The specific values are summarized in Table \ref{tab:parameters}.
    \item \textbf{Communication system}: In this experiment, we assume the same quasi-static channel for the wireless environment as mentioned in Section \ref{sec:numerical_example}. Similarly, the physical layer complies with the 3GPP standard. For simplicity, we only adopt modulation during the transmission process, including BPSK, QPSK, 8QAM, 16QAM, 32QAM, and 64QAM. 
    To align with modern communication systems, the differentiable quantization technique is also adopted to enable training with modulation. As for the wireless inference, we normalize the output features of the encoder and quantize the floating-point precision features into 1-bit fixed-point binary representations.
\end{enumerate}

\begin{table}[htbp]
  \caption{Learning rate and regularization parameter}
  \label{tab:parameters}
  \centering
  \begin{tabular}{| m{0.2\columnwidth}<{\centering} | m{0.3\columnwidth}<{\centering} | m{0.1\columnwidth}<{\centering} | m{0.1\columnwidth}<{\centering} |}
    \hline
    Dataset & Model & $\eta$ & $\beta$ \\ \hline
    Cats \& Dogs & 6-layer CNN & 1e-3 & 0.5 \\ \hline
    CIFAR-10 & VGG11 & 1e-1 & 0.001 \\ \hline
    CIFAR-10 & ResNet18 &  1e-1 & 0.1 \\ \hline
    CIFAR-100 & VGG16 & 1e-1 & 0.001 \\ \hline
    CIFAR-100 & ResNet34 & 1e-1 & 0.5 \\ \hline
    Tiny-ImageNet & ViT-small-patch16-224 & 1e-3 & 0.1 \\ \hline
  \end{tabular}
\end{table}

We first jointly train the encoder and the decoder on the corresponding dataset under standard distributed learning. Next, we fine-tune the model using the proposed robust training framework in Algorithm \ref{alg:SGLD}, vanilla Adam, and the MI method in wireless distributed learning where the SNR of the channel is 10 dB and the modulation scheme is QPSK. 
During the training process, we calculate and record the mutual information $I(W,Z;S)$ in the derived upper bound using Algorithm \ref{alg:approximation}. To demonstrate the superiority of our proposed robust training framework in terms of robustness against channel distortions, we compare the inference accuracy of the trained models in wireless distributed learning. Specifically, we conduct wireless distributed learning in a channel at an SNR of 10 dB using multiple transmission rates, achieved by employing different modulation schemes.

\begin{figure}[htbp]
  \centerline{\subfigure[]{\includegraphics[width=0.5\columnwidth]{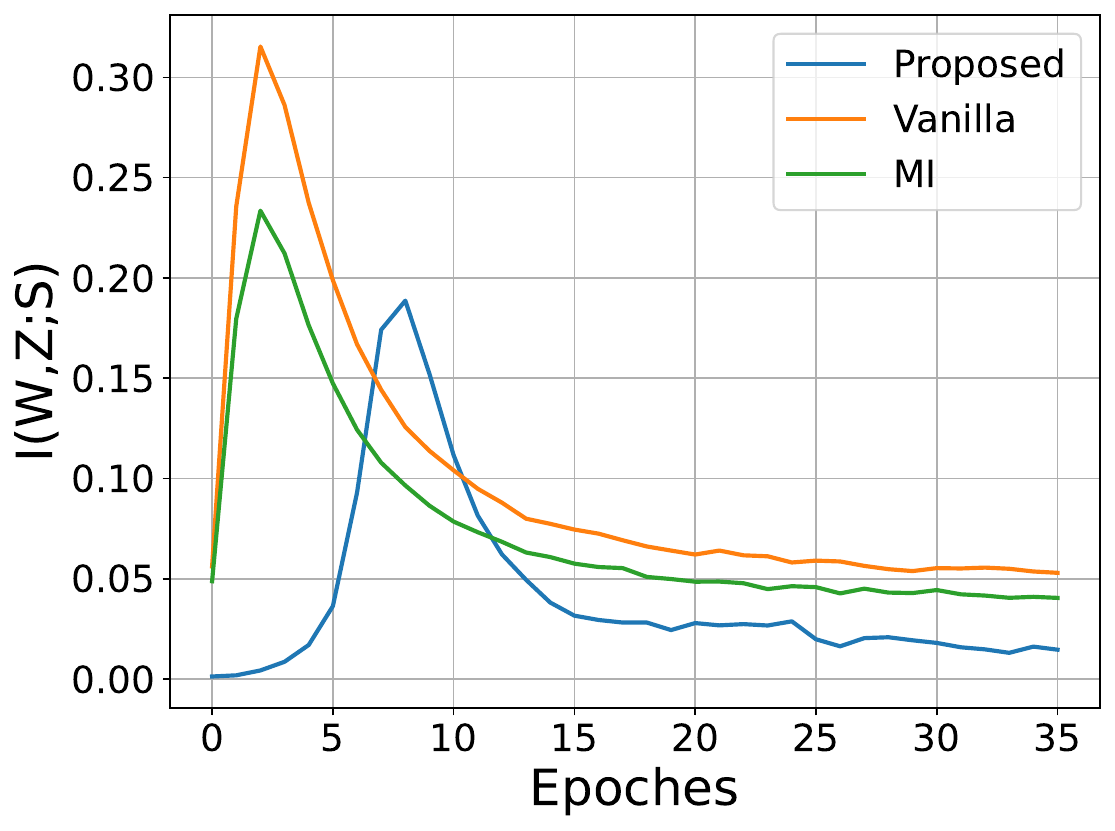}
      \label{fig:catdog_info}}
    \subfigure[]{\includegraphics[width=0.5\columnwidth]{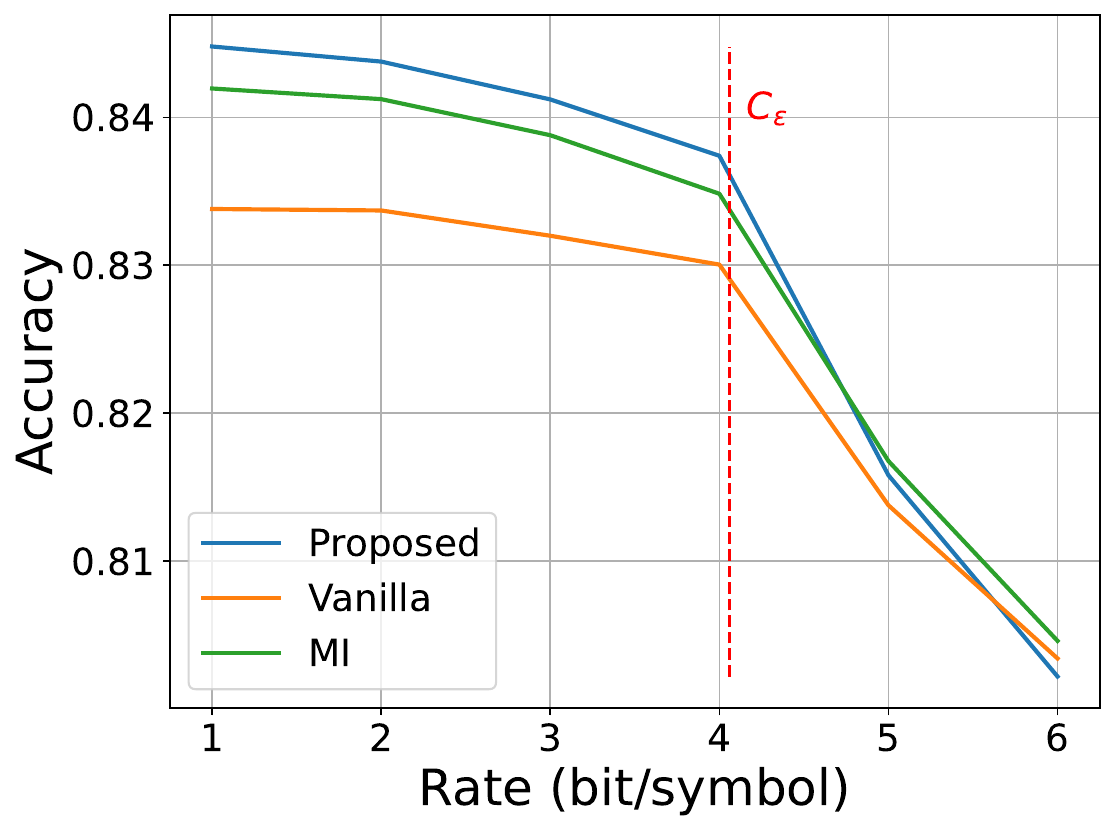}
      \label{fig:catdog_acc_rate}}}
  \caption{6-layer CNN on Cats \& Dogs Dataset. \textbf{(a)}: The approximated mutual information $\Tilde{I}(W,Z;S)$ in the upper bound during the training. \textbf{(b)}: The accuracy on the test dataset w.r.t the transmission rates.}
  \label{fig:catdog}
\end{figure}

\begin{figure}[htbp]
  \centerline{\subfigure[]{\includegraphics[width=0.5\columnwidth]{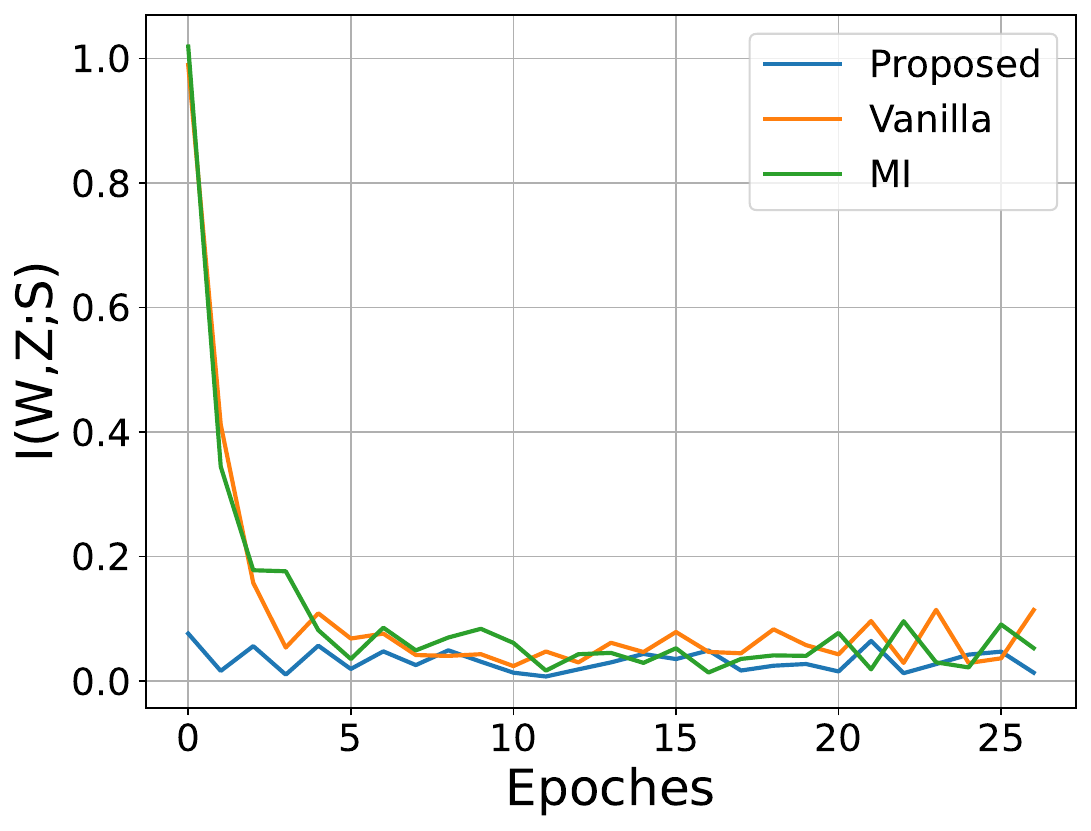}
      \label{fig:cifar_vgg_info}}
    \subfigure[]{\includegraphics[width=0.5\columnwidth]{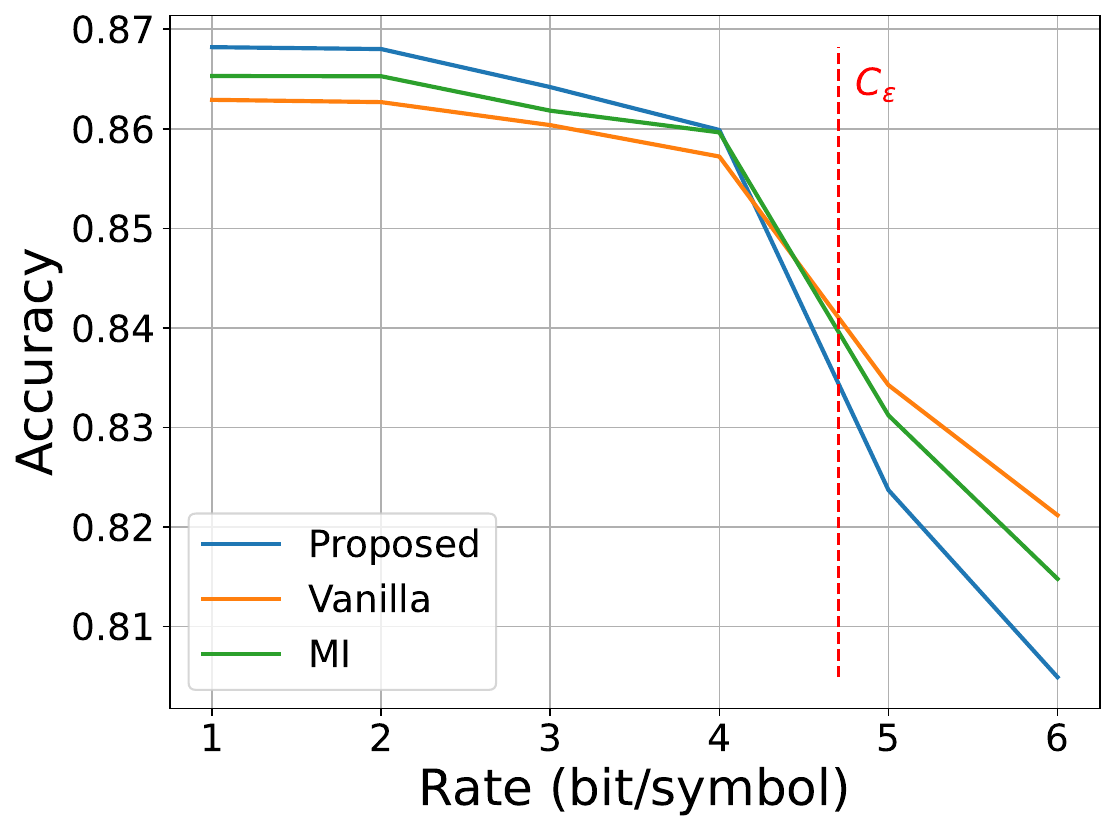}
      \label{fig:cifar_vgg_acc_rate}}}
  \caption{VGG11 on CIFAR-10 Dataset. \textbf{(a)}: The approximated mutual information $\Tilde{I}(W,Z;S)$ in the upper bound during the training. \textbf{(b)}: The accuracy on the test dataset w.r.t the transmission rates.}
  \label{fig:cifar_vgg}
\end{figure}

\begin{figure}[htbp]
  \centerline{\subfigure[]{\includegraphics[width=0.5\columnwidth]{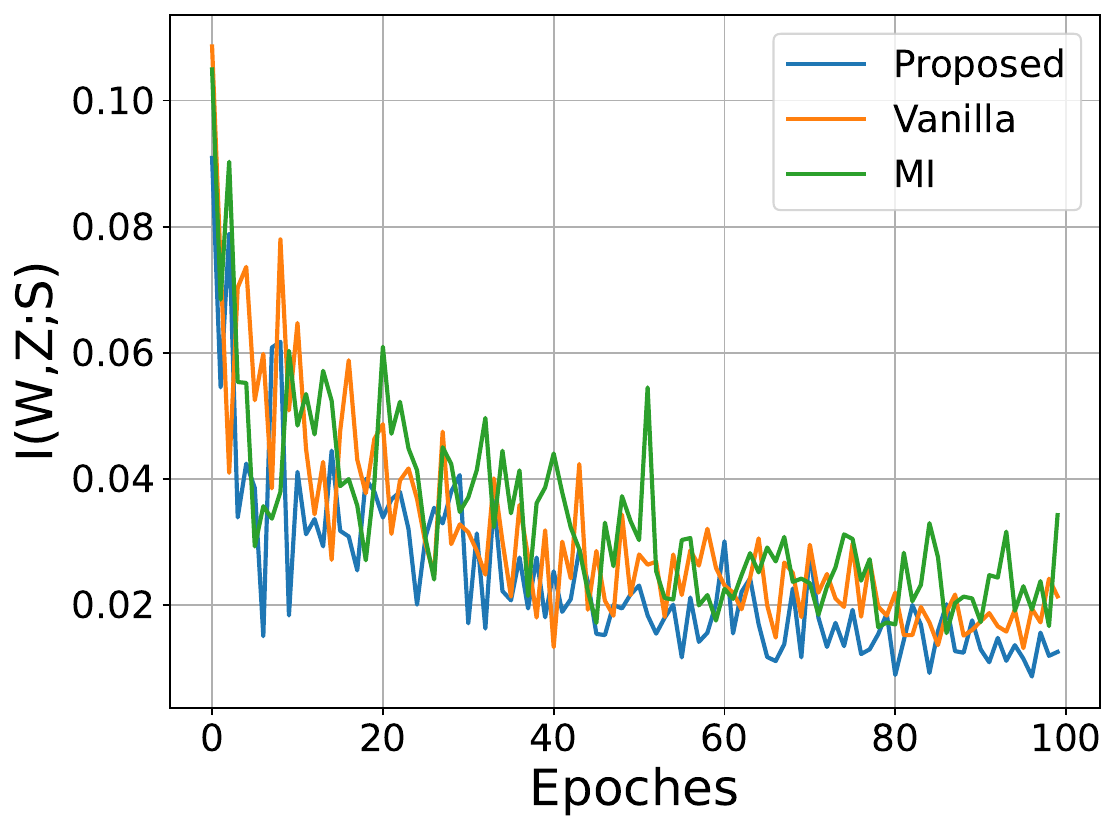}
      \label{fig:cifar_resnet_info}}
    \subfigure[]{\includegraphics[width=0.5\columnwidth]{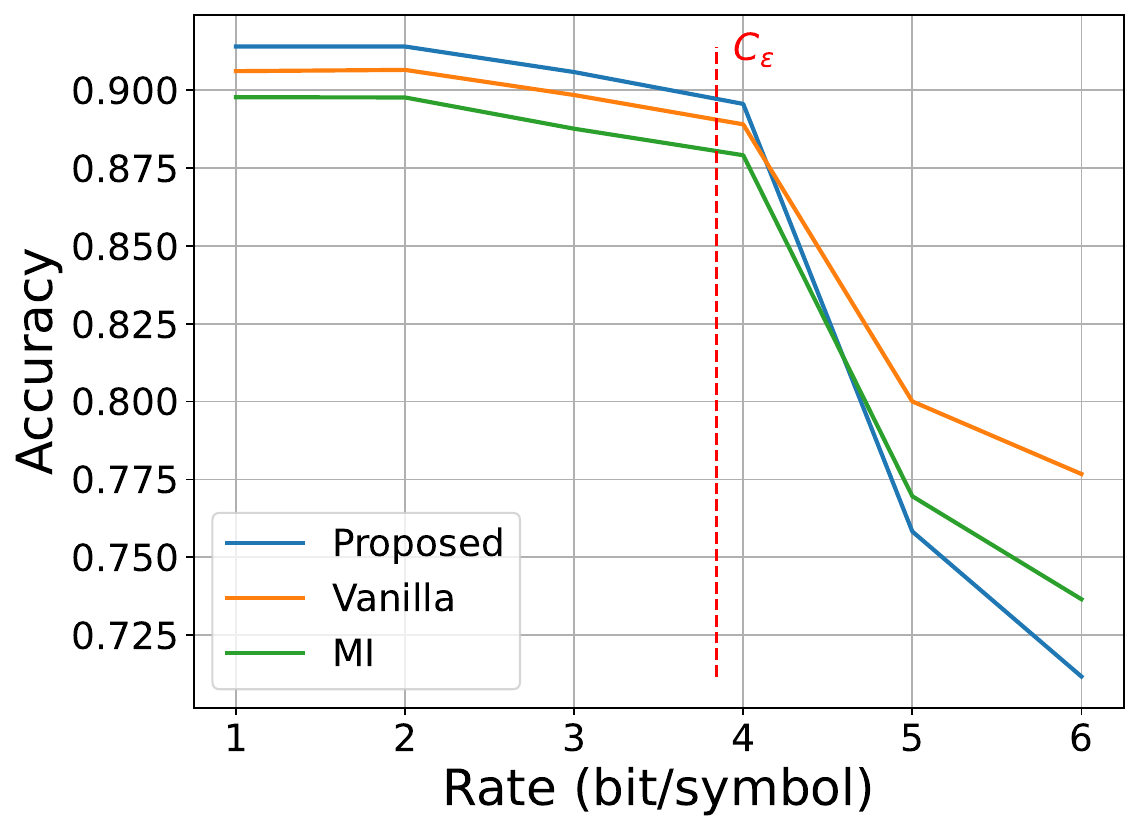}
      \label{fig:cifar_resnet_acc_rate}}}
  \caption{ResNet18 on CIFAR-10 Dataset. \textbf{(a)}: The approximated mutual information $\Tilde{I}(W,Z;S)$ in the upper bound during the training. \textbf{(b)}: The accuracy on the test dataset w.r.t the transmission rates.}
  \label{fig:cifar_resnet}
\end{figure}

\begin{figure}[htbp]
  \centerline{\subfigure[]{\includegraphics[width=0.5\columnwidth]{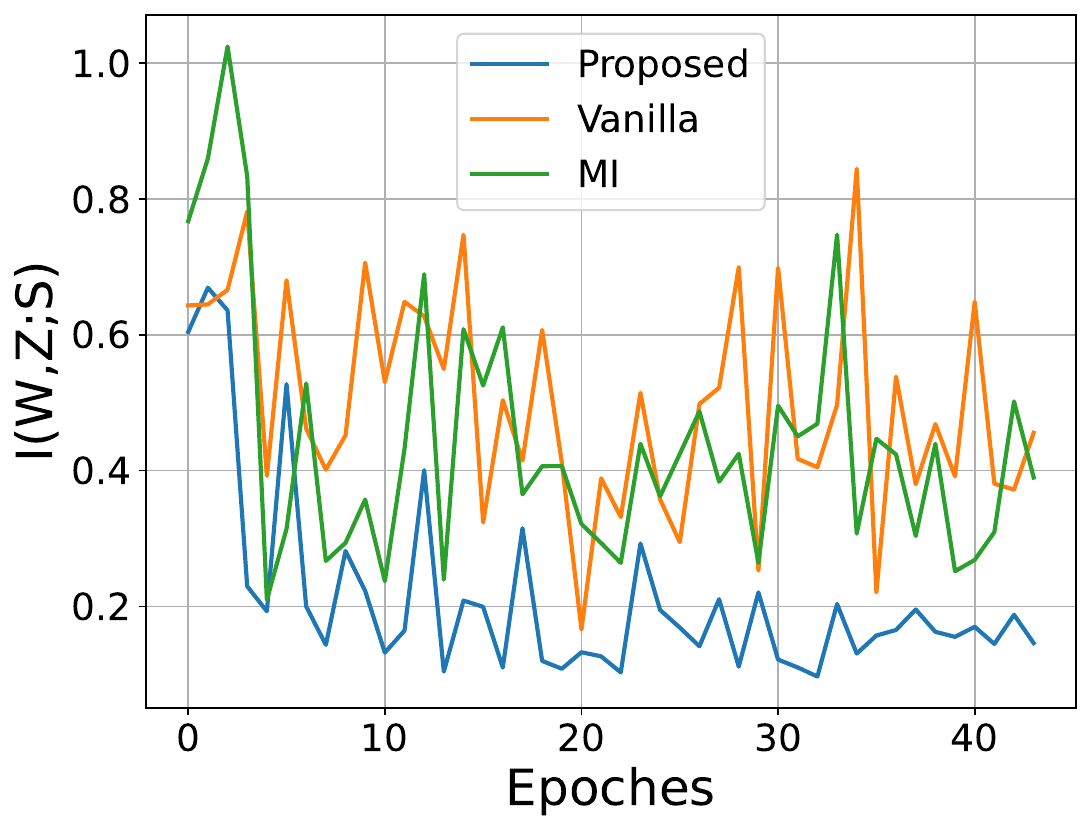}
      \label{fig:cifar100_vgg_info}}
    \subfigure[]{\includegraphics[width=0.5\columnwidth]{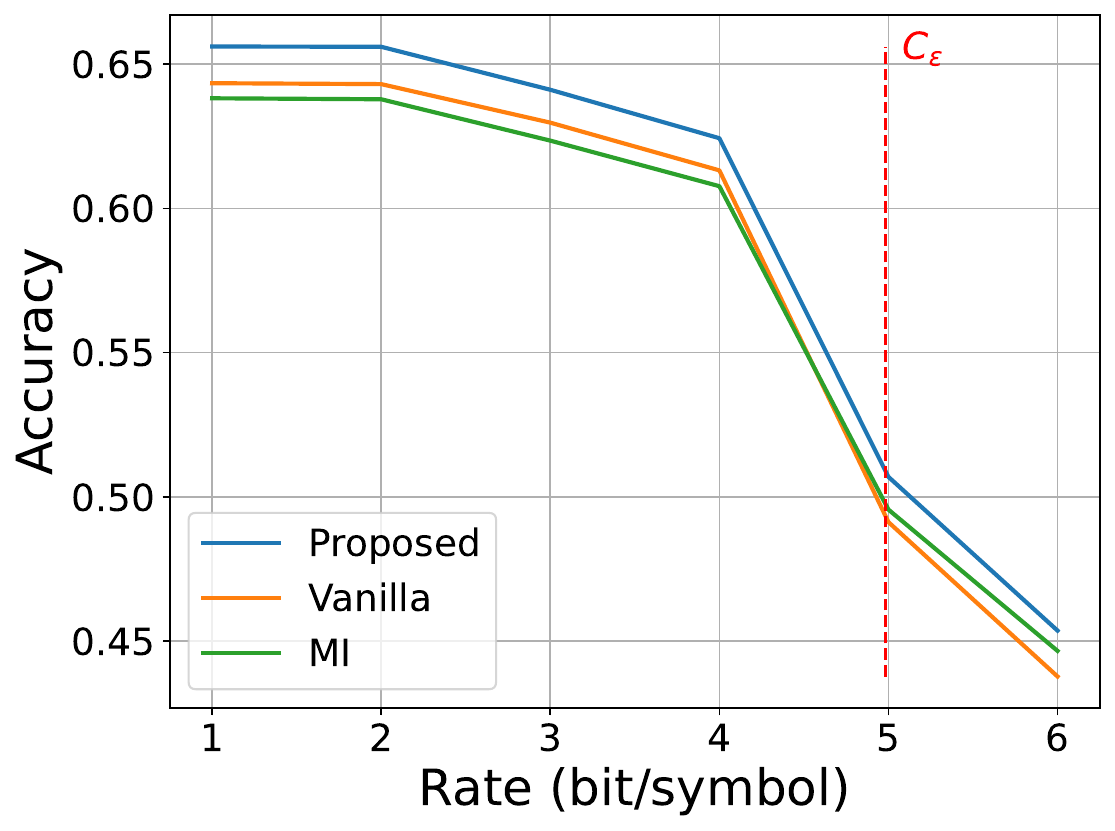}
      \label{fig:cifar100_vgg_acc_rate}}}
  \caption{VGG16 on CIFAR-100 Dataset. \textbf{(a)}: The approximated mutual information $\Tilde{I}(W,Z;S)$ in the upper bound during the training. \textbf{(b)}: The accuracy on the test dataset w.r.t the transmission rates.}
  \label{fig:cifar100_vgg}
\end{figure}

\begin{figure}[htbp]
  \centerline{\subfigure[]{\includegraphics[width=0.5\columnwidth]{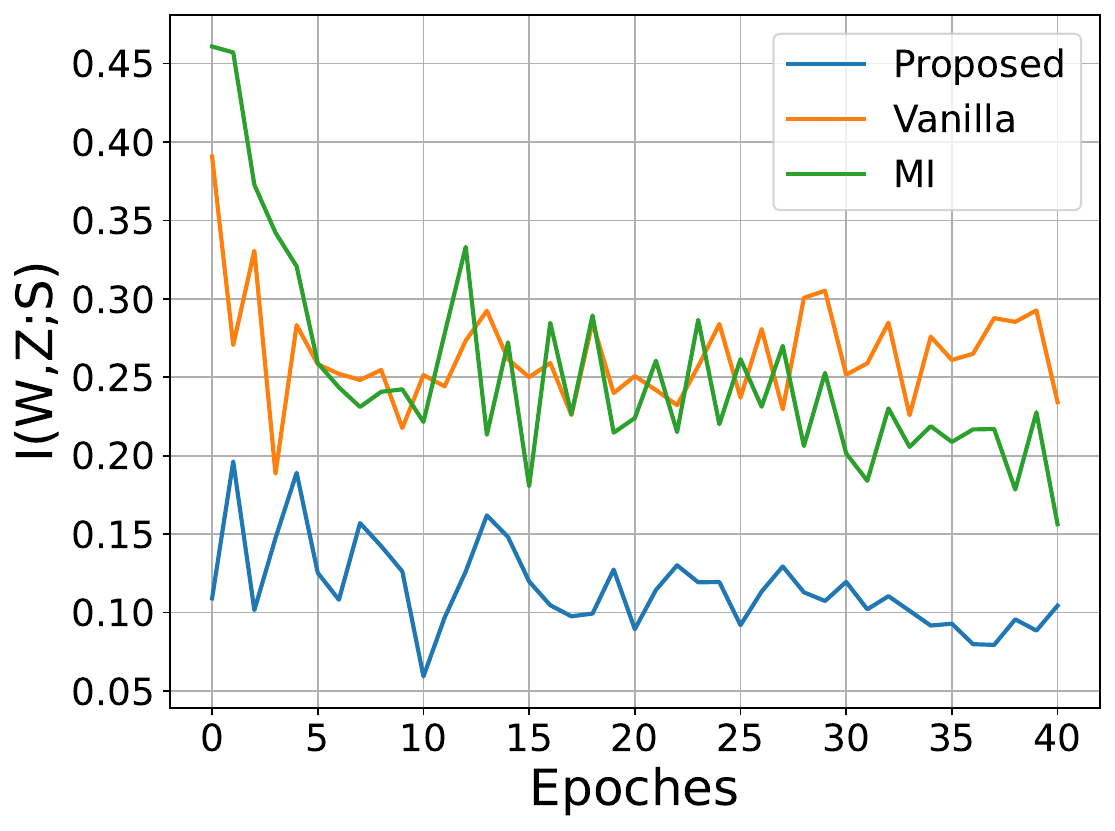}
      \label{fig:cifar100_resnet_info}}
    \subfigure[]{\includegraphics[width=0.5\columnwidth]{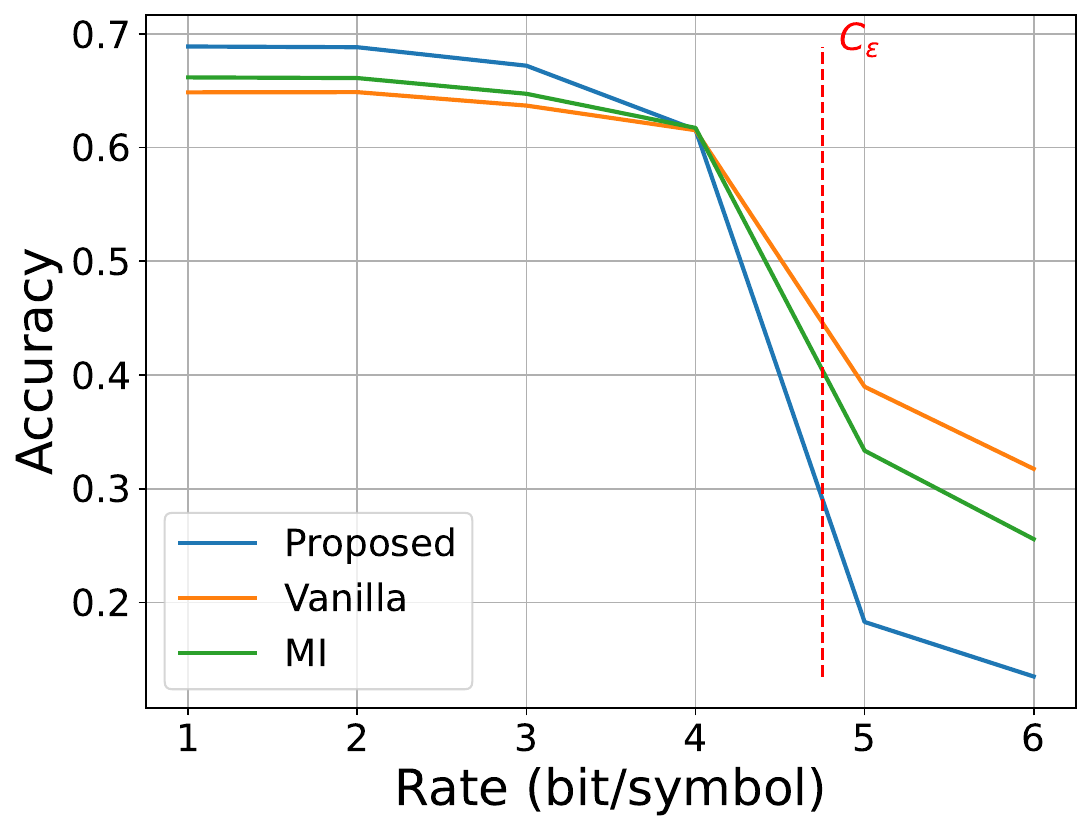}
      \label{fig:cifar100_resnet_acc_rate}}}
  \caption{ResNet34 on CIFAR-100 Dataset. \textbf{(a)}: The approximated mutual information $\Tilde{I}(W,Z;S)$ in the upper bound during the training. \textbf{(b)}: The accuracy on the test dataset w.r.t the transmission rates.}
  \label{fig:cifar100_resnet}
\end{figure}

\begin{figure}[htbp]
  \centerline{\subfigure[]{\includegraphics[width=0.5\columnwidth]{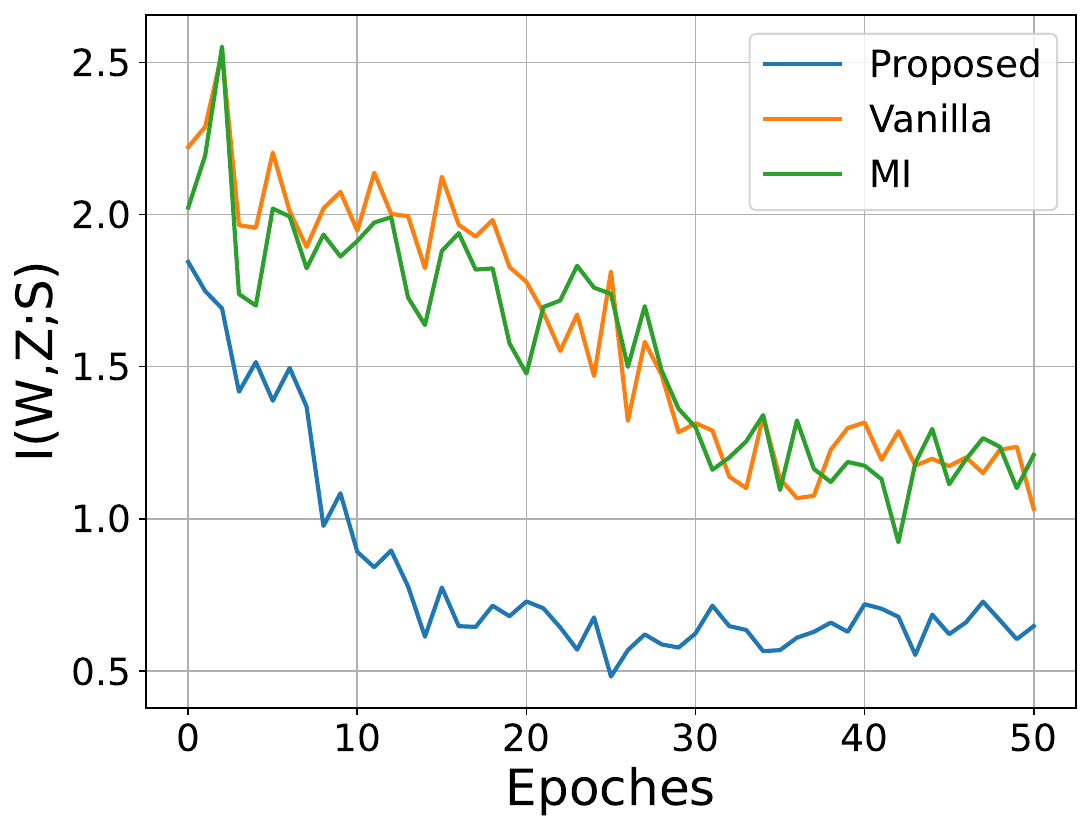}
      \label{fig:vit_info}}
    \subfigure[]{\includegraphics[width=0.5\columnwidth]{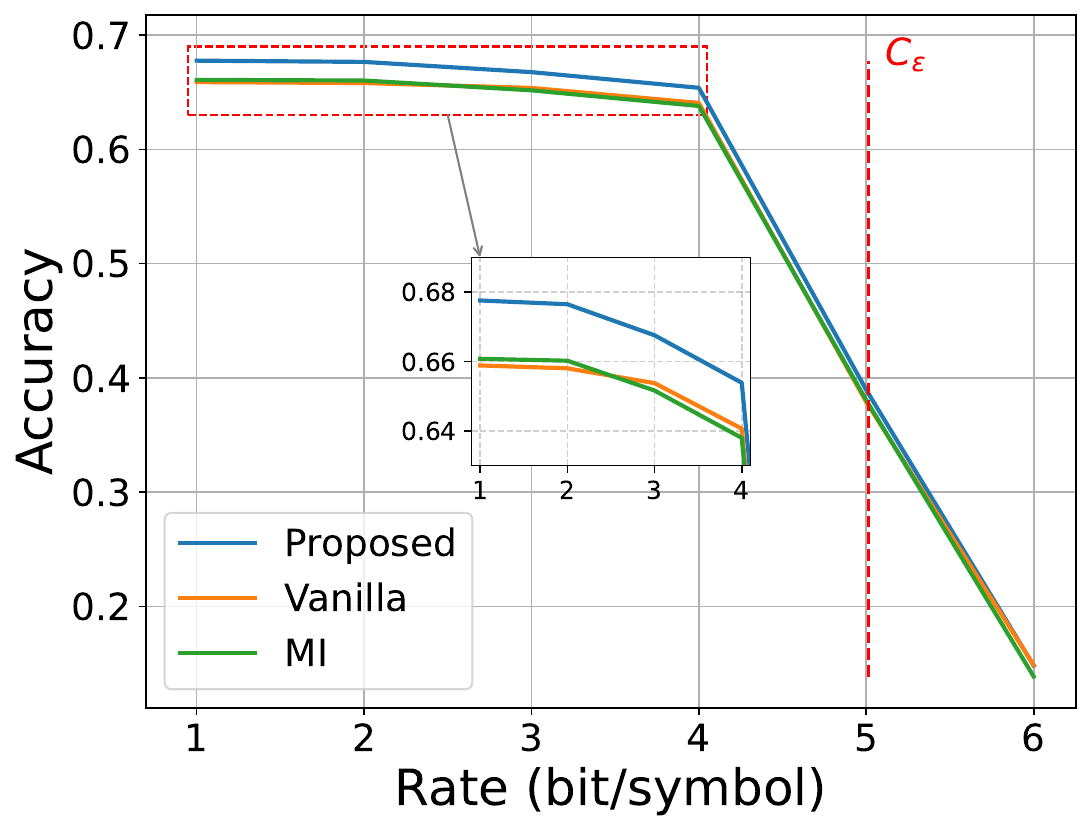}
      \label{fig:vit_acc_rate}}}
  \caption{ViT on Tiny-ImageNet Dataset. \textbf{(a)}: The approximated mutual information $\Tilde{I}(W,Z;S)$ in the upper bound during the training. \textbf{(b)}: The accuracy on the test dataset w.r.t the transmission rates.}
  \label{fig:vit}
\end{figure}

\subsection{Robustness Evaluation}
The training and inference results are summarized in Figures \ref{fig:catdog}-\ref{fig:vit}. 
It can be seen from the figures that the mutual information $I(W,Z;S)$ decreases over epochs and converges to a small value. This is because the NN's weights are continuously updated throughout the training process, causing the mutual information to change. It eventually stabilizes as the optimal weights of the NN are learned. This suggests that the training process enhances the inherent robustness of the NN against wireless distortion.
It can be seen from the figures that the proposed robust training framework achieves a smaller upper bound $G=\sigma\sqrt{2I(W,Z;S)}$ during training and thus improves the robustness of wireless distributed learning. 
Additionally, we can see from the result of the simple CNN in Figure \ref{fig:catdog_info} that the mutual information $I(W,Z;S)$ surges in the first few training epochs, and then decreases slowly approaching zero. Such a phenomenon corresponds to the 'fitting' and 'compression' hypotheses in the IB theory of deep learning \cite{IB-DNN}. Specifically, the model learns information about the task in the fitting phase and then increases the generalization performance in the compression phase. 
The training process of the robust training framework could also be divided into two phases according to the curve of mutual information $I(W,Z;S)$.
In the first phase, $I(W,Z;S)$ increases since the model extracts all relevant information from the wireless channel and the wireless inference task. In the second phase, the model starts to abandon the irrelevant information and regularize against the wireless discrepancy gap between the wireless risk and standard risk. However, for deeper and more complex NNs, the first fitting phase might be too short and might not be observed in the experiments.

Comparing the test accuracy during the inference process, the model trained in the robust training framework outperforms the baselines. We credit the improvement to the explicit consideration of information regularization during the training process, which enhances the model's generalization ability under various channel conditions.
In our experiment, the channel capacity at an SNR of 10 dB is approximately 3.46 bit/symbol. 
Constrained by Shannon's channel capacity, the conventional communication system can achieve error-free transmission at rates up to 3 bit/symbol. 
Meanwhile, we statistically compute the task outage probability $p_e$. By taking $\epsilon=p_e$, we calculate the task-aware $\epsilon$-capacity in \cref{equ:max_rate}, which is noted in the figures.
The simulation results indicate that the proposed framework achieves higher task accuracy up to the rate of 4 bit/symbol, enabling the learning model to transmit at a rate beyond the Shannon capacity while maintaining inference accuracy. It demonstrates that we could leverage the intrinsic robustness within wireless distributed learning for higher communication efficiency while maintaining the inference performance.
Therefore, the proposed robust training framework can effectively enhance the robustness of wireless distributed learning, thereby exploiting this robustness to transcend the Shannon capacity.
However, it is worth noting that the task accuracy for all three methods declines rapidly at high transmission rates above 4 bit/symbol. In some cases, the proposed framework may even underperform compared to the baselines. This is because the transmission rate exceeds the proposed task-aware $\epsilon$-capacity, leading to a high task outage probability. Such a scenario is beyond the scope of this study, which primarily focuses on theoretically exploring the task-aware transmission upper limit that slightly exceeds the Shannon capacity by leveraging the robustness in wireless distributed learning.

\subsection{The Impact of Prior Distribution}In our theoretical analysis, the joint distribution $p(w,z)$ characterizes the relation between model weights $w$ and sample $z$ under standard distributed learning. We note that in standard distributed learning, the learning algorithm $p(w|z)$ aims to find the optimal weights minimizing the loss function. However, it is almost impossible to achieve $\E_{p(w,z)}\left[l(W,Z)\right]$ in practical simulations as existing learning models inevitably exhibit residual error. 

In our simulations, we first train the model under standard distributed learning to learn $\theta_z$ that parameterizes the distribution $p(w|z)$. This is considered as the prior weights and used to estimate the mutual information $I(W,Z;S)$ as shown in Algorithm \ref{alg:approximation}. Therefore, it is crucial to optimize $p(w,z)$, as a better prior leads to a more accurate mutual information estimation. In return, the proposed robust training framework can more effectively enhance the robustness of wireless distributed learning. We conduct the following simulations to demonstrate the impact of $p(w,z)$ and the results are summarized in Figure \ref{fig:priorepoch}. Specifically, we pre-train the model using a varying number of epochs under standard distributed learning to obtain $\theta_z$. Under wireless distributed learning, we utilize the proposed robust training framework by regularizing with mutual information calculated using $\theta_z$, and the trained model is deployed for inference in a wireless channel at an SNR of 10 dB using multiple transmission rates. It can be seen from the figures that longer pre-training improves the accuracy within the task-achievable region. 
This aligns with our theoretical insight that an optimal $p(w,z)$ tightens the wireless discrepancy upper bound, ensuring the robustness of wireless distributed learning.

\begin{figure}[htbp]
  \centerline{\subfigure[]{\includegraphics[width=0.5\columnwidth]{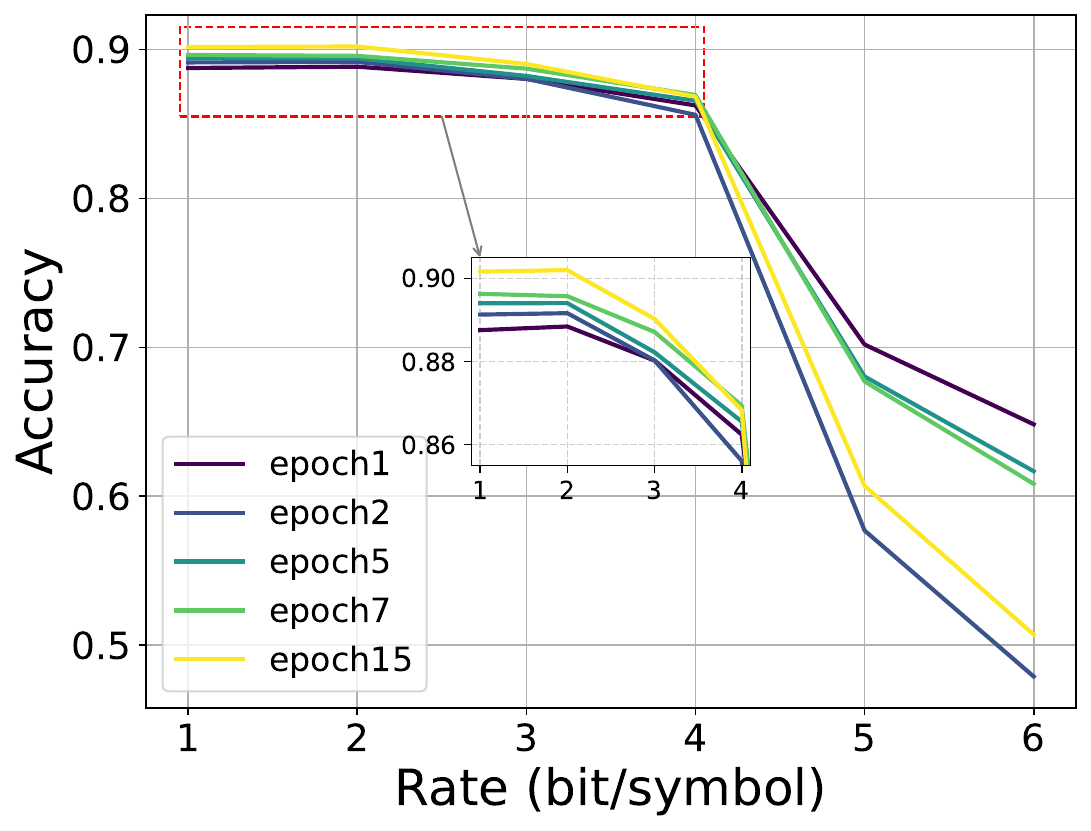}
      \label{fig:cifar_resnet_acc_vs_rate_priorepoch}}
    \subfigure[]{\includegraphics[width=0.5\columnwidth]{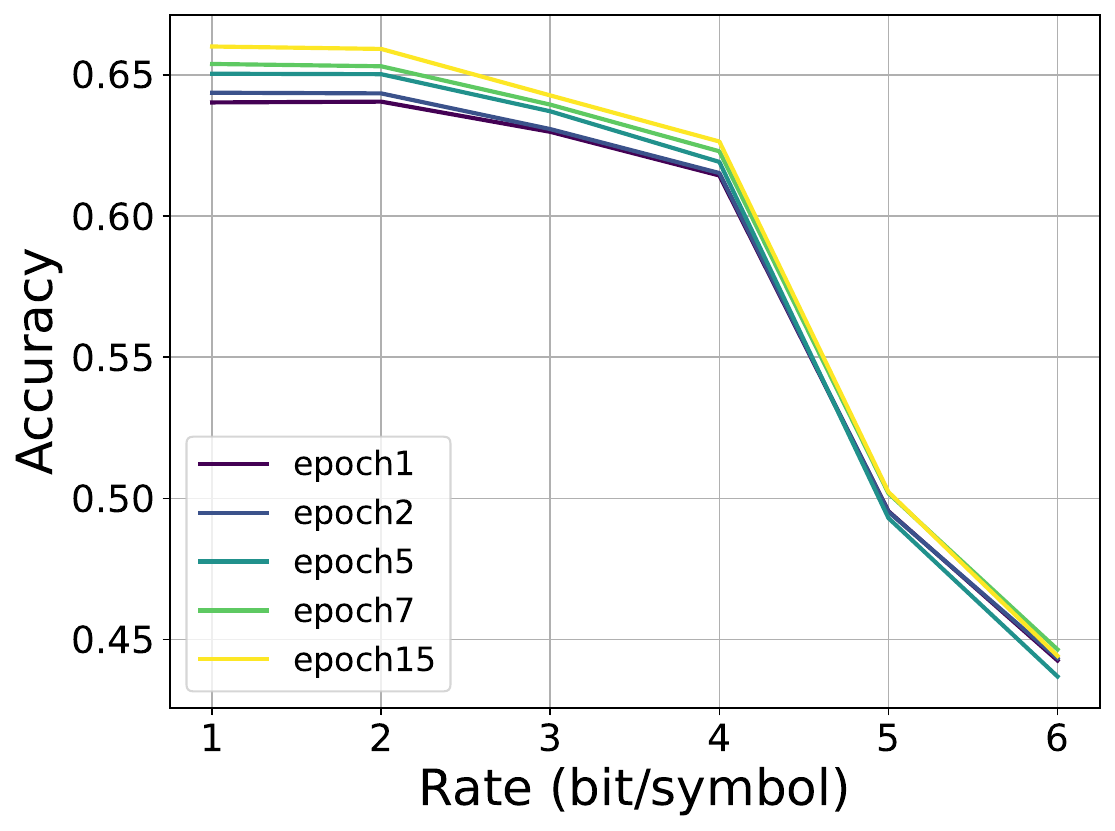}
      \label{fig:cifar100_vgg_acc_vs_rate_priorepoch}}}
  \caption{The accuracy on the test dataset w.r.t. the transmission rates. \textbf{(a)}: ResNet18 on CIFAR-10 Dataset. \textbf{(b)}: VGG16 on CIFAR-100 Dataset.}
  \label{fig:priorepoch}
\end{figure}

\section{Conclusion}
\label{sec:conclusion}
In this paper, we have presented an information-theoretic analysis of the robustness within wireless distributed learning and its implications for achieving task-aware $\epsilon$-capacity. By leveraging information theory, we have derived an upper bound in terms of mutual information on the performance deterioration due to wireless fading. The upper bound serves as a measure of the robustness against channel distortions.
In contrast to the bit-level metric, we have defined task outage based on the upper bound as a metric for successful inference.
Thereby, we have characterized the maximum achievable rate in wireless distributed learning and the task-aware $\epsilon$-capacity benefiting from robustness.
We have further proposed an efficient algorithm to estimate the upper bound in practice.
Given the theoretical analysis of the robustness in wireless distributed learning, we have designed a robust training framework based on the derived upper bound, which is capable of improving the model's robustness.
The robustness in wireless distributed learning and the effectiveness of our framework were validated through numerical experiments.

While our work provides novel insights into evaluating the robustness in wireless distributed learning, several limitations present opportunities for future exploration. Our current analysis assumes model hypotheses $W$ depend on channel side information $S$, which may not apply to pre-trained models where this side information is unavailable during training. Additionally, although we proved achievability for task-aware $\epsilon$-capacity in Theorem \ref{thm:max_rate}, developing a tight converse remains challenging. This would require new coding strategies or NN architectures. Future work could also focus on deriving tighter bounds by applying more refined inequalities or using conditions in specific scenarios. Finally, our definition of task outage based on loss exceeding the upper bound could be supplemented with metrics that more directly relate to task accuracy, as increased loss does not always indicate incorrect inference. Addressing these limitations could further improve robustness in wireless distributed learning systems and enhance the task-aware channel capacity.

\renewcommand{\arraystretch}{1.2}
\begin{table*}[htbp]
    \small
    \caption{NN architecture in wireless distributed learning}
    \label{tab:architecture}
    \centering
    \begin{center}
        \begin{threeparttable}
            \begin{tabular}{|m{0.2\textwidth}<{\centering}|m{0.35\textwidth}<{\centering}|m{0.35\textwidth}<{\centering}|}
                \hline
                Model & encoder & decoder \\
                \hline
                6-layer CNN & Conv16\textsuperscript{\textcolor{gray}{*1}} + Conv32 + Conv64 & FC500\textsuperscript{\textcolor{gray}{*2}} + FC50 + FC2 \\ 
                \hline
                VGG11      & Conv64 + Conv128 + Conv256$\times$2 + Conv512$\times$2 & Conv512$\times$2 + FC4096$\times 2$ + FC10 \\ 
                \hline
                VGG16      & Conv64$\times$2 + Conv128$\times$2 + Conv256$\times$3 & Conv512$\times$3 + Conv512$\times$3 + FC4096$\times 2$ + FC100 \\ 
                \hline
                ResNet18      & Conv64 + Block64$\times$2\textsuperscript{\textcolor{gray}{*3}} + Block128$\times$2 & Block256$\times$2 + Block512$\times$2 + FC10 \\ 
                \hline
                ResNet34      & Conv64 + Block64$\times$3 + Block128$\times$4 & Block256$\times$6 + Block512$\times$3 + FC10 \\ 
                \hline
                ViT-small-patch16-224      & PatchEmbed + Transformer384$\times$6\textsuperscript{\textcolor{gray}{*4}} & Transformer384$\times$6 + FC200 \\ 
                \hline
            \end{tabular}
            \begin{tablenotes}
                    \item \textcolor{gray}{*1} \textcolor{gray}{Conv represents a convolutional layer with kernel size 3$\times$3, and 16 represents the number of channels.} 
                    \item \textcolor{gray}{*2} \textcolor{gray}{FC represents a fully connected layer, and 500 represents the number of nodes in the fully connected layer.}
                    \item \textcolor{gray}{*3} \textcolor{gray}{Block represents a residual block with two convolutional layers and a skip connection, and 64 represents the number of channels of the convolutional layers.}
                    \item \textcolor{gray}{*4} \textcolor{gray}{Transformer represents a residual Transformer encoder with the number of heads 12, and 384 represents the dimension of the embedding.}
            \end{tablenotes}
        \end{threeparttable}
    \end{center}
\end{table*}

\appendices

\section{Proof of Corollary \ref{thm:upper_bound_gamma}}
\label{apx:proof_gamma}
\begin{proof}
\label{thm:proof_gamma}
    Leveraging the Donsker-Varadhan variational representation in \cref{equ:KL_divergence}, we have
    \begin{equation}
        \label{equ:KL_divergence_gap_gamma}
        \begin{split}
            & D_{\text{KL}}(q(w, z | s)\| p(w,z)) \\ 
            & \geq \lambda\left(\E_{q}\left[ l(W,Z)\right] - \E_{p}\left[ l(W,Z)\right]\right) - \frac{\lambda^2\sigma^2}{2(1-c\lambda)}. 
        \end{split}
    \end{equation}
    For simple notation, we denote $D_{\text{KL}}(q(w, z | s)\| p(w,z))$ as $D(q\|p)$. Following the same proof in Theorem \ref{thm:upper_bound} that the discriminant must be non-positive, we have
    \begin{equation}
        \left[cD(q\|p) - \left(\E_{q}\left[ l(W,Z)\right] - \E_{p}\left[ l(W,Z)\right]\right)\right]^2 \leq 2\sigma^2 D(q\|p).
    \end{equation}
    Therefore, the difference between expectations can be upper-bounded as follows:
    \begin{equation}
        \E_{q}\left[ l(W,Z)\right] - \E_{p}\left[ l(W,Z)\right] \leq \sigma\sqrt{2D(q\|p)} + cD(q\|p).
    \end{equation}
    Taking the expectation under the side information distribution $\mu$ and utilizing Jensen's inequality gives
    \begin{equation}
        \label{equ:jensen_inequality_gamma}
        \begin{split}
            & \E_{\mu}\left[\E_q[l(W,Z)]-\E_p[l(W,Z)]\right] \\
            & \leq \E_{\mu}\left[\sigma\sqrt{2D(q\|p)} + cD(q\|p)\right] \\
            & \leq \sqrt{2 \sigma^2 \E_{\mu}\left[D(q\|p)\right]} + cI(W,Z;S)\\
            & = \sigma\sqrt{2I(W,Z;S)} + cI(W,Z;S).
        \end{split}
    \end{equation}
    This completes the proof. 
    
\end{proof}

\section{Experiment settings}
\label{sec:architecture}
In wireless distributed learning, the NN model $f^{(W)}$ is separated into an encoder $f_e^{(W)}$ at the UE and a decoder $f_d^{(W)}$ at the BS. The detailed deployment of the models in the experiment, including VGG11, VGG16, ResNet18, ResNet34, and ViT, is summarized in Table \ref{tab:architecture}.


\end{document}